\documentclass[11pt]{article}

\usepackage{subcaption}

\usepackage[normalem]{ulem}

\usepackage{nicematrix}
\providecommand{\U}[1]{\protect\rule{.1in}{.1in}}
\usepackage{amsmath,amssymb,amsthm}
\usepackage{amsfonts}

\renewcommand\and{\end{tabular}\kern-\tabcolsep\ and\ \kern-\tabcolsep\begin{tabular}[t]{c}}
\let\origthanks\thanks
\renewcommand\thanks[1]{\begingroup\let\rlap\relax\origthanks{#1}\endgroup}

\usepackage{autonum}

\newcommand{\quotes}[1]{``#1''}

\newtheorem{theorem}{Theorem}
\newtheorem{corollary}{Corollary}
\newtheorem{definition}{Definition}
\newtheorem{example}{Example}
\newtheorem{lemma}{Lemma}
\newtheorem{assumption}{Assumption}
\newtheorem{problem}{Problem}

\makeatletter \def\epstopdf@sys@cmd{epstopdf} \makeatother

\begin{document}

\title{Clustering for epidemics on networks: a geometric approach}
\author{Bastian Prasse\thanks{Faculty of Electrical Engineering, Mathematics and
Computer Science, P.O Box 5031, 2600 GA Delft, The Netherlands; \emph{email}:
b.prasse@tudelft.nl, p.f.a.vanmieghem@tudelft.nl}, Karel Devriendt\thanks{Mathematical Institute, University of Oxford, Oxford UK; 
\emph{email}: devriendt@maths.ox.ac.uk; Also at Alan Turing Institute, London, UK} and Piet Van Mieghem\footnotemark[1]}
\maketitle

\begin{abstract}
Infectious diseases typically spread over a contact network with millions of individuals, whose sheer size is a tremendous challenge to analysing and controlling an epidemic outbreak. For some contact networks, it is possible to group individuals into clusters. A high-level description of the epidemic between a few clusters is considerably simpler than on an individual level. However, to cluster individuals, most studies rely on equitable partitions, a rather restrictive structural property of the contact network. In this work, we focus on Susceptible-Infected-Susceptible (SIS) epidemics, and our contribution is threefold. First, we propose a geometric approach to specify \emph{all} networks for which an epidemic outbreak simplifies to the interaction of only a few clusters. Second, for the complete graph and \textit{any} initial viral state vectors, we derive the closed-form solution of the nonlinear differential equations of the $N$-Intertwined Mean-Field Approximation (NIMFA) of the SIS process. Third, by relaxing the notion of equitable partitions, we derive low-complexity approximations and bounds for epidemics on \emph{arbitrary} contact networks. Our results are an important step towards understanding and controlling epidemics on large networks.
\end{abstract}

\section{Introduction}

Modern epidemiology encompasses a broad range of spreading phenomena \cite{pastor2015epidemic,nowzari2016analysis,kiss2017mathematics}. The majority of viruses spread through a population of tremendous size, which renders individual-based modelling impractical. However, most applications do not require to model an epidemic on individual level. Instead, a mesoscale description of the epidemic often is sufficient. For instance, suppose the outbreak of a virus is modelled on the level of neighbourhoods. Then, sophisticated lockdown measures can be deployed which constrain neighbourhoods differently, depending on the prevalence of the virus in the respective neighbourhood. The natural way to obtain a mesoscale description of the epidemic is \textit{clustering} (or grouping) of individuals, for instance, by assigning individuals with similar age or location to the same cluster. Thus, all individuals in one cluster are considered indistinguishable and exchangeable. Additionally to the complexity reduction, clustering for epidemics on networks has the advantage that, on a mesoscale description, temporal fluctuations of the individual-based contact network may average out.

We consider a contact network with $N$ nodes. Every node $i=1, ..., N$ corresponds to an individual or a group of individuals. We focus on the Susceptible-Infected-Susceptible (SIS) epidemic process in an individual-based mean-field approximation, where every node $i$ has a viral state $v_i(t)\in [0,1]$ at every time $t$. The evolution of the viral state $v_i(t)$ is governed by a set of $N$ nonlinear differential equations:

\begin{definition}[NIMFA \cite{lajmanovich1976deterministic, van2009virus, mieghem2014homogeneous}]
 For every node $i$, the viral state $v_i(t)$ evolves in continuous time $t\ge 0$ as 
\begin{align}
	\frac{d v_i (t)}{d t } & = - \delta_i v_i(t) + \left( 1 - v_i(t)\right) \sum^N_{j = 1} \beta_{i j}  v_j(t), \label{NIMFA_continuous}
\end{align}
where $\delta_i >0$ is the \emph{curing rate} of node $i$, and $\beta_{i j} > 0$ is the \emph{infection rate} from node $j$ to $i$.  
\end{definition}

If the nodes correspond to individuals, then the differential equations (\ref{NIMFA_continuous}) follow from a mean-field approximation of the stochastic SIS process \cite{van2009virus, van2011n}, and the viral state $v_i(t)$ approximates the expected value $\operatorname{E}[ X_i(t)]$ of the zero-one state $X_i(t)$ of the stochastic SIS process. For a zero-one, or Bernoulli, random variable the expectation $\operatorname{E}[ X_i(t)]$ is equal to the probability $\operatorname{Pr}[ X_i(t)=1]$ that node $i$ is infected at time $t$. In the remainder of this work, we refer to (\ref{NIMFA_continuous}) as NIMFA, which stands for \quotes{$N$-Intertwined Mean-Field Approximation} \cite{van2009virus, van2011n}. The advantage of NIMFA is that the SIS Markov chain with $2^N$ states is approximated by $N$ nonlinear differential equations. NIMFA follows from the SIS process by the approximation $\operatorname{E}[ X_i(t)X_j(t)]\approx \operatorname{E}[ X_i(t)]\operatorname{E}[ X_j(t)]$. Around the epidemic threshold, the approximation of the stochastic SIS process by NIMFA might be inaccurate \cite{van2009virus}. Furthermore, we stress that NIMFA (\ref{NIMFA_continuous}) assumes that the viral dynamics are Markovian and that the infection rates $\beta_{ij}$ do not depend on time $t$. Markovian and non-Markovian viral dynamics can be substantially different \cite{van2013non}. 

The contact network, assumed to be fixed and time-invariant, corresponds to the $N\times N$ infection rate matrix $B$, which is composed of the elements $\beta_{ij}$. We denote by $\operatorname{diag}(x)$ the $N\times N$ diagonal matrix with the vector components of $x\in \mathbb{R}^N$ on its diagonal. We denote the $N\times N$ curing rate matrix $S = \textrm{diag}(\delta_1, ..., \delta_N)$. Then, the matrix representation of NIMFA (\ref{NIMFA_continuous}) is
\begin{align}\label{NIMFA_stacked}
	\frac{d v (t)}{d t } & = - S  v(t) + \textup{\textrm{diag}}\left( u - v(t)\right) B  v(t),
\end{align}
where $v(t) = (v_1(t), ..., v_N(t))^T$ is the viral state vector at time $t$, and $u$ is the $N\times 1$ all-one vector. \textit{Homogeneous} NIMFA \cite{van2009virus} assumes the same infection rate $\beta$ and curing rate $\delta$ for all nodes,
\begin{align}
	\frac{d v (t)}{d t } & = - \delta  v(t) + \beta \textup{\textrm{diag}}\left( u - v(t)\right) A  v(t),\label{NIMFA_continuous_homogeneous}
\end{align}
where $A$ is an $N\times N$ zero-one adjacency matrix.

For NIMFA (\ref{NIMFA_continuous}), the basic reproduction number $R_0$ follows \cite{van2002reproduction} as
\begin{align}\label{R_0_NIMFA_def}
R_0 = \rho(S^{-1}B), 
\end{align}
 where $\rho(M)$ denotes the spectral radius of a square matrix $M$. Around the \textit{epidemic threshold} $R_0$, there is a bifurcation \cite{lajmanovich1976deterministic}. If $R_0\le 1$, then the all-healthy state, $v_i(t)=0$ for all nodes $i$, is the only equilibrium of NIMFA (\ref{NIMFA_stacked}), and it holds that $v(t)\rightarrow 0$ as $t\rightarrow \infty$. If $R_0>1$, then there is a second equilibrium, the \textit{steady-state} vector $v_\infty$, with positive components, and it holds that $v(t)\rightarrow v_\infty$ as $t\rightarrow \infty$, if $v(0)\neq 0$. 
 
Many papers deal with clustering of individuals into \textit{communities} \cite{clauset2008hierarchical,abbe2017community,peixoto2014hierarchical}, where individuals within the same community are densely connected, and there are only few links between individuals of different communities. Hence, communities are defined by structural properties of the contact graph. Most results are of the type: if the network has a certain mesoscale structure, then also the dynamics have some structure \cite{arenas2006synchronization,o2013observability, bonaccorsi2015epidemic}. In this work, we approach clustering from the other direction: we presume structure \textit{in the dynamics} and aim to find all contact networks that are compatible with the structured dynamics. 

The central analysis tool in our analysis is the \textit{proper orthogonal decomposition} (POD) \cite{brunton2019data} of the $N\times 1$ viral state vector $v(t)$, which is given by 
\begin{align}\label{POD_exact}
v(t) = \sum^m_{l=1} c_l(t) y_l
\end{align}
for some $m\le N$. Here, the $N\times 1$ \textit{agitation mode} vectors~$y_1, ..., y_m$ are orthonormal\footnote{A set of vectors $y_1, ..., y_m$ is orthonormal if $y^T_l y_k=0$ for $l\neq k$ and  $y^T_l y_k=1$ for $l= k$.}, and the scalar functions $c_l(t)=y^T_lv(t)$ are obtained by projecting the viral state~$v(t)$ onto the vector $y_l$. Since any $N\times 1$ vector $v(t)$ can be written as the linear combination of $N$ orthonormal vectors, the POD (\ref{POD_exact}) is exact for \textit{any} network if $m=N$. However, we are particularly interested in networks, for which the number of agitation modes $m$ is (much) smaller than the number of nodes $N$. If (\ref{POD_exact}) holds true, then the viral state vector $v(t)$ is element of the $m$ dimensional subspace
\begin{align}\label{V_invariant_span}
\mathcal{V} =\operatorname{span}\{y_1, ..., y_m\}
\end{align}
at any time $t$, where the span (the set of all linear combinations) of the vectors $y_1, ..., y_m$ is denoted by
\begin{align}
\operatorname{span}\left\{y_1, ..., y_m\right\} = \left\{ 
\sum^m_{l=1} c_l y_l \Big| c_l \in \mathbb{R}\right\}.
\end{align}
With the POD (\ref{POD_exact}), the viral state $v(t)$ can be described with less than $N$ differential equations: denote the right side of the NIMFA (\ref{NIMFA_stacked}) by $f_\text{NIMFA}\left( v(t) \right)\in\mathbb{R}^N$. Then, NIMFA (\ref{NIMFA_stacked}) reads more compactly
\begin{align}\label{oijlikjnljaergrdegaerg}
\frac{d v(t)}{dt} =  f_\text{NIMFA}\left( v(t) \right).
\end{align}
With the POD (\ref{POD_exact}), we obtain that
\begin{align} \label{jknaergdrgergregerager}
\sum^m_{l=1} \frac{d c_l(t)}{dt}  y_l = f_\text{NIMFA}\left( \sum^m_{l=1} c_l(t)  y_l \right).
\end{align}
Since the vectors $y_1, ..., y_m$ are orthonormal, we can project (\ref{jknaergdrgergregerager}) onto the agitation modes $y_l$ to obtain the differential equations
\begin{align} \label{aergtgrtedgtasgtsrht}
\frac{d c_l(t)}{dt} = y^T_l f_\text{NIMFA}\left( \sum^m_{l=1} c_l(t) y_l \right), \quad l=1, ..., m.
\end{align}
Hence, the POD (\ref{POD_exact}) reduces the number of differential equations from the number of nodes $N$ to the number of agitation modes $m$. We emphasise that the POD (\ref{POD_exact}) is a hybrid of linear and nonlinear analysis: The viral state $v(t)$ equals a linear combination of the agitation modes $y_l$, which are weighted by possibly nonlinear functions $c_l(t)$. In \cite{prasse2020predicting}, we have shown that the POD (\ref{POD_exact}) is an accurate approximation for a diverse class of dynamics on networks. In this work, we study under which conditions the POD (\ref{POD_exact}) is \textit{exact} for the NIMFA epidemic model (\ref{NIMFA_stacked}). 
 
\begin{example}\label{example:path_graph}
\begin{figure}[!ht]
    \centering
    \begin{minipage}[b]{.4\textwidth}
\centering
\subfloat[Path graph.\label{fig:chapter2_fig_1a}]{\includegraphics[width=0.15\textwidth]{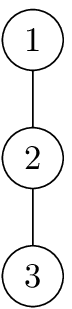} \quad \quad \quad \quad}
\end{minipage}
    \centering
    \begin{minipage}[b]{.4\textwidth}
\centering
\subfloat[Viral state space.\label{fig:chapter2_fig_1b}]{\includegraphics[width=\textwidth]{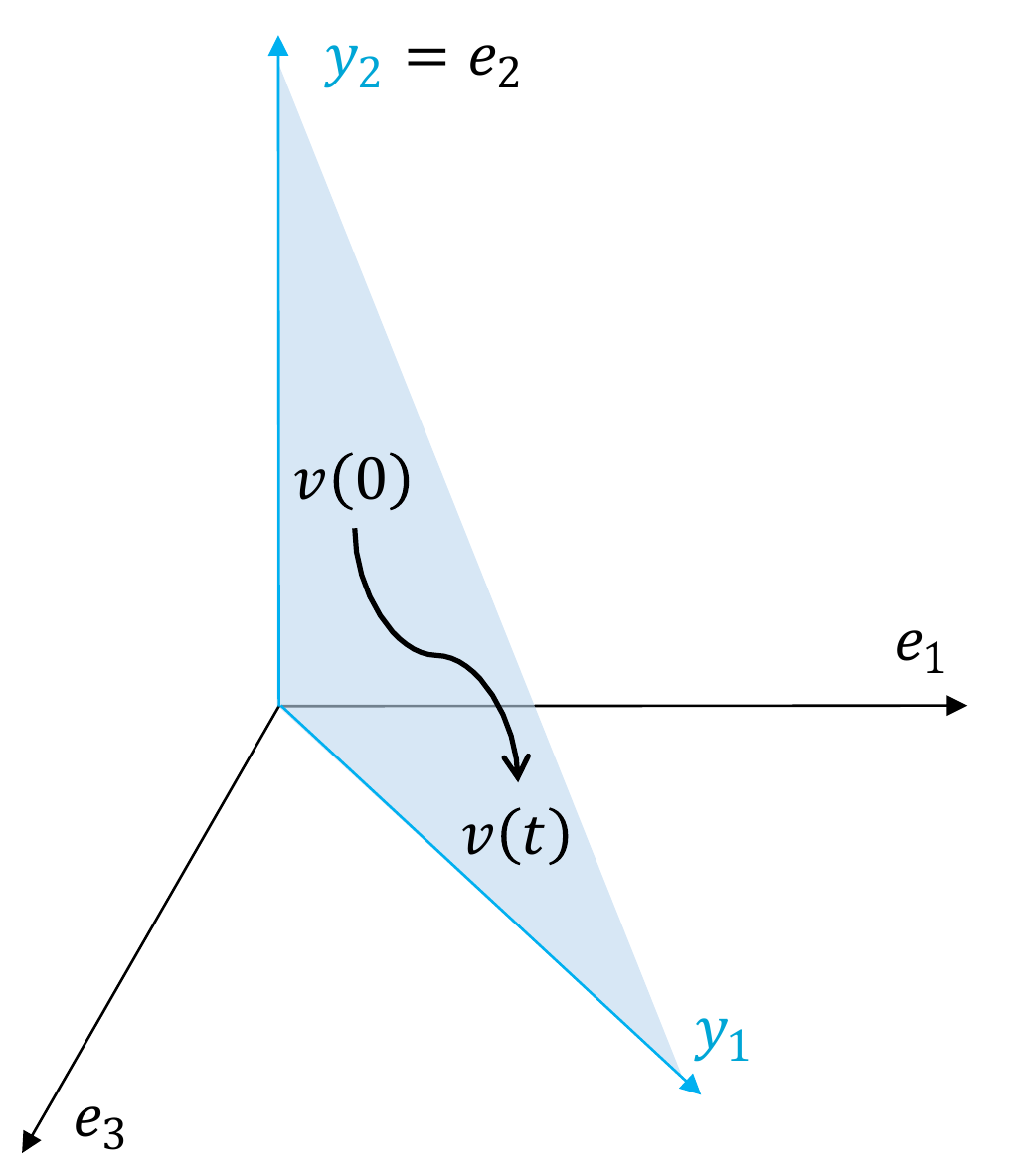}}
\end{minipage}        
    \caption{\textbf{Proper orthogonal decomposition for a path graph.} \textbf{(a)}: A path graph with $N=3$ nodes. The top, middle and bottom nodes are labelled by 1, 2 and 3, respectively. \textbf{(b)}: The black curve depicts the trajectory of the viral state $v(t)$ in the Euclidean space $\mathbb{R}^3$. The shaded area illustrates the viral state set $\mathcal{V}$, which equals the span of the vectors $y_1, y_2$, given by (\ref{yaetvfdvdf}). Provided that $v(0)\in\mathcal{V}$, the viral state $v(t)$ remains in the subspace $\mathcal{V}$ at every time $t$.}\label{fig:chapter2_fig_1}
\end{figure}
Consider homogeneous NIMFA (\ref{NIMFA_continuous_homogeneous}) on the path graph in Figure~\ref{fig:chapter2_fig_1a}, for which the viral state vector $v(t)$ evolves as
\begin{align}\label{ESFSERFascsdcv}
\frac{d v_1 (t)}{d t}  &= - \delta  v_1(t) + \beta \left( 1- v_1(t)\right) v_2(t), \\
\frac{d v_2 (t)}{d t}  &= - \delta  v_2(t) + \beta \left( 1- v_2(t)\right) \left( v_1(t) + v_3(t)\right), \\
\frac{d v_3 (t)}{d t}  &= - \delta  v_3(t) + \beta \left( 1- v_3(t)\right) v_2(t).
\end{align} 
 Suppose that the initial viral states of node 1 and 3 are equal, $v_1(0)=v_3(0)$. Then, it holds that $v_1(t)=v_3(t)$ at all times~$t$ due to the symmetry of the path graph. Hence, the viral state vector $v(t)=(v_1(t), v_2(t), v_3(t))^T$ satisfies
\begin{align} \label{eargawrfawefwef}
v(t) = c_1(t) y_1 + c_2(t) y_2,
\end{align}
where the orthonormal vectors $y_1, y_2$ are given by
\begin{align}\label{yaetvfdvdf}
y_1 = \frac{1}{\sqrt{2}}\begin{pmatrix}
1 \\
0 \\
1
\end{pmatrix}, \quad 
y_2 = \begin{pmatrix}
0 \\
1 \\
0
\end{pmatrix}.
\end{align}
As illustrated by Figure~\ref{fig:chapter2_fig_1b}, the viral state $v(t)$ remains in the $m=2$ dimensional subspace $\mathcal{V} =\operatorname{span}\{y_1, y_2\}$ at all times $t$, provided that $v(0)\in\mathcal{V}$. On the subspace $\mathcal{V}$, (\ref{aergtgrtedgtasgtsrht}) yields that the $N=3$ differential equations (\ref{ESFSERFascsdcv}) reduce to $m=2$ equations
\begin{align}
\frac{d c_1 (t)}{d t} &= -\delta c_1(t) + \sqrt{2}\beta \left( 1 - \frac{1}{\sqrt{2}}c_1(t)\right) c_2(t),\\
\frac{d c_2 (t)}{d t} &= -\delta c_2(t) + 2 \sqrt{2} \beta \left( 1 - c_2(t)\right) c_1(t),
\end{align}
from which the viral state $v(t)$ is obtained with (\ref{eargawrfawefwef}).
\end{example}

Two conditions must hold for the set~$\mathcal{V}$ to reduce NIMFA to $m$ differential equations. First, the set~$\mathcal{V}$ must be an $m$ dimensional subspace, spanned by the basis vectors $y_1, ..., y_m$. Second, if the initial viral state $v(0)$ is element of the set $\mathcal{V}$, then the viral state $v(t)$ must remain in the set $\mathcal{V}$ at every time $t>0$. Hence, the set $\mathcal{V}$ must be an \textit{invariant set} of NIMFA. Thus, we consider the geometric problem:
\begin{problem}[Clustering in NIMFA]\label{problem:clustering}
For a given number of nodes $N$ and a given number $m\le N$ of agitation modes, find \emph{all} $N\times N$ infection rate matrices~$B$ and the corresponding $N\times 1$ agitation modes $y_1, ..., y_m$, such that $\mathcal{V}=\operatorname{span}\{y_1, ..., y_m\}$ is an invariant set of NIMFA (\ref{NIMFA_stacked}).
\end{problem}

In contrast to Example~\ref{example:path_graph}, for which the agitation modes $y_1, y_2$ follow rather straightforwardly, Problem~\ref{problem:clustering} considers the interdependency of arbitrary graphs and invariant sets $\mathcal{V}$ in full generality. 

If $m<<N$, then we expect that the invariant set $\mathcal{V}$, and its basis vectors $y_l$, reflect a macroscopic structure, or a clustering, of the contact graph. For instance, the agitation mode $y_1$ in Example~\ref{example:path_graph} indicates that the viral states $v_1(t)$ and $v_3(t)$ evolve equally and nodes 1 and 3 can be assigned to the same cluster. 

Furthermore, the invariant set $\mathcal{V}$ allows for sophisticated, low-complexity control methods for the viral state $v(t)$, see \cite{nowzari2016analysis} for a survey of control methods. More specifically, consider that an affine control method is applied to NIMFA (\ref{oijlikjnljaergrdegaerg}),
\begin{align}\label{kjnbaerfgergfreg}
\frac{d v(t)}{dt} =  f_\text{NIMFA}\left( v(t) \right) + \sum^m_{l=1} g_l(t) y_l.
\end{align}
Here, the scalar function $g_l(t)$ is the control of the $l$-th agitation mode $y_l$. If the subspace $\mathcal{V}=\operatorname{span}\{y_1, ..., y_m\}$ is an invariant set of NIMFA (\ref{NIMFA_stacked}), then $\mathcal{V}$ is also an invariant set of (\ref{kjnbaerfgergfreg}). Hence, on the subspace $\mathcal{V}$, the viral state $v(t)$ can be controlled with only $m$ distinct control inputs $g_1(t), ..., g_m(t)$. If the agitation mode $y_l$ corresponds to a group of nodes, such as in Example~\ref{example:path_graph}, then the control $g_l(t)$ is applied to all nodes of that group. For instance, $g_l(t)$ could be the viral state control of individuals of a certain age group and location. 

\section{Related work}
Clustering in NIMFA is closely related to \textit{equitable partitions} \cite{schwenk1974computing, van2010graph, schaub2020hierarchical}. We denote a general partition of the node set $\mathcal{N}=\{1, ..., N\}$ by\footnote{Slightly deviating from common notation, we also refer to $\pi$ as an (equitable) partition \textit{of the infection rate matrix $B$}.} $\pi = \{\mathcal{N}_1, ..., \mathcal{N}_r\}$. Here, the \textit{cells} $\mathcal{N}_1, ..., \mathcal{N}_r$ are disjoint subsets of the node set $\mathcal{N}$, such that $\mathcal{N}=\mathcal{N}_1 \cup...\cup \mathcal{N}_r$. We adapt the definition of \textit{equitable} partitions in \cite{mugnolo2014semigroup, ottaviano2018optimal} as:
\begin{definition}[Equitable partition] \label{def:equitable_partition}
Consider a symmetric $N\times N$ infection rate matrix~$B$ and a partition $\pi = \{\mathcal{N}_1, ..., \mathcal{N}_r\}$ of the node set $\mathcal{N}=\{1, ..., N\}$. The partition $\pi$ is \emph{equitable} if, for all cells $l,p=1, ..., r$, the infection rates $\beta_{ik}$ satisfy
\begin{align}
\sum_{k \in \mathcal{N}_l} \beta_{ik}= \sum_{k \in \mathcal{N}_l} \beta_{jk} \quad \forall i,j\in\mathcal{N}_p.
\end{align}
\end{definition}
  
For an equitable partition $\pi$, we define the degree from cell $\mathcal{N}_l$ to cell $\mathcal{N}_p$ as
\begin{align}\label{khbrffdg}
d_{pl}=\sum_{k \in \mathcal{N}_l} \beta_{ik}
\end{align}
for some node $i\in\mathcal{N}_p$. Definition~\ref{def:equitable_partition} states that, for an equitable partition $\pi$, the sum of the infection rates (\ref{khbrffdg}) is the same for all nodes $i\in \mathcal{N}_p$. We denote the $r\times r$ \textit{quotient matrix} by $B^{\pi}$, whose elements are defined as $\left(B^{\pi}\right)_{pl}= d_{pl}$. Furthermore, we define the $r\times 1$ all-one vector $u_r =(1, ..., 1)^T$. 

As shown by Bonaccorsi \textit{et al.} \cite{bonaccorsi2015epidemic} and Ottaviano \textit{et al.} \cite{ottaviano2018optimal}, NIMFA (\ref{NIMFA_stacked}) can be reduced to $r$ differential equations, provided that the infection rate matrix~$B$ has an equitable partition $\pi$ with $r$ cells. For our work, we summarise the results in \cite{bonaccorsi2015epidemic, ottaviano2018optimal} as:

\begin{theorem}[\cite{bonaccorsi2015epidemic, ottaviano2018optimal}] \label{theorem:equitableOriginal}
Consider NIMFA (\ref{NIMFA_stacked}) on an $N\times N$ infection rate matrix $B$ with an equitable partition $\pi= \{\mathcal{N}_1, ..., \mathcal{N}_r\}$. Assume that $\delta_i=\delta_j$ and $v_i(0)=v_j(0)$ for all nodes $i,j$ in the same cell $\mathcal{N}_l$. Then, it holds that $v_i(t)=v_j(t)$ at every time $t>0$ for all nodes $i,j\in\mathcal{N}_l$ and all $l=1, ...,r$. Furthermore, define the $r\times 1$ \emph{reduced-size} viral state vector $v^{\pi}(t) = \left( v_{i_1}(t), ..., v_{i_r}(t) \right)^T$ and the $r\times r$ \emph{reduced-size} curing rate matrix
\begin{align}\label{definition_S_pi}
S^{\pi} = \operatorname{diag}\left( \delta_{i_1}, ..., \delta_{i_r} \right),
\end{align}
where $i_l$ denotes an arbitrary node in the cell $\mathcal{N}_l$. Then, the reduced-size viral state vector $v^{\pi}(t)$ evolves as
\begin{align}\label{kjnkjnertgverg}
	\frac{d v^{\pi}(t)}{d t } & = - S^{\pi}  v^{\pi}(t) + \textup{\textrm{diag}}\left( u_r - v^{\pi}(t)\right) B^{\pi}  v^{\pi}(t).
\end{align}
\end{theorem}
Remarkably, on both microscopic (\ref{NIMFA_stacked}) and macroscopic (\ref{kjnkjnertgverg}) resolutions, the viral dynamics follow the same class of governing equation. For the Markovian Susceptible-Infectious-Susceptible (SIS) process, Simon et \textit{al.} \cite{simon2011exact} proposed a lumping approach to reduce the complexity, which is an approximation and merges states of the SIS Markov chain, also see the work of Ward \textit{et al.} \cite{ward2019exact}. In \cite{devriendt2017unified}, a generalised mean-field framework for Markovian SIS epidemics has been proposed, which includes NIMFA as a special case. Beyond epidemics, analogous results to Theorem~\ref{theorem:equitableOriginal} have been proved for a diverse set of dynamics\footnote{Specifically, we believe that Theorem~\ref{theorem:equitableOriginal} can be generalised to the dynamics $\frac{d v_i(t)}{dt} = -\delta_i v_i(t) +\sum^N_{j=1} \beta_{ij} g(v_i(t), v_j(t))$, where the arbitrary function $g(v_i(t), v_j(t))$ describes the \quotes{coupling} \cite{timme2007revealing, barzel2013universality, laurence2019spectral, prasse2020predicting} between node $i$ and $j$.} on networks with equitable partitions \cite{egerstedt2012interacting, o2013observability, pecora2014cluster, schaub2016graph, devriendt2020bifurcation}. As a direct consequence of Theorem~\ref{theorem:equitableOriginal}, equitable partitions are related to the proper orthogonal decomposition (\ref{POD_exact}):
\begin{corollary}\label{corollary:equitable_are_invariant}
Consider NIMFA (\ref{NIMFA_stacked}) on an $N\times N$ infection rate matrix $B$ with an equitable partition $\pi= \{\mathcal{N}_1, ..., \mathcal{N}_r\}$. Assume that $\delta_i=\delta_j$ and $v_i(0)=v_j(0)$ for all nodes $i,j$ in the same cell $\mathcal{N}_l$. Then, the subspace $\mathcal{V}=\operatorname{span}\{y_1, ..., y_m\}$ with $m=r$ is an invariant set, where the $N\times 1$ agitation modes $y_l$ are given by
\begin{align} \label{khjbsenrggaserdtgttr}
\left( y_l\right)_i=\begin{cases}
\frac{1}{\sqrt{\left| \mathcal{N}_l \right|}} \quad &\text{if} \quad i\in\mathcal{N}_l,\\
0 &\text{if} \quad i\not\in\mathcal{N}_l,
\end{cases}
\end{align}
and the scalar functions equal $c_l(t)= \sqrt{\left| \mathcal{N}_l \right|} v^{\pi}_l(t)$.
\end{corollary}
In other words, Corollary~\ref{corollary:equitable_are_invariant} states that every equitable partition $\pi$ yields an invariant set $\mathcal{V}$, whose dimension equals the number of cells $r$ in the partition $\pi$. Example~\ref{example:equitable_partition} illustrates Theorem~\ref{theorem:equitableOriginal} and Corollary~\ref{corollary:equitable_are_invariant}:
\begin{example}\label{example:equitable_partition}
\begin{figure}[!ht]
    \centering
    \includegraphics[width=0.7\textwidth]{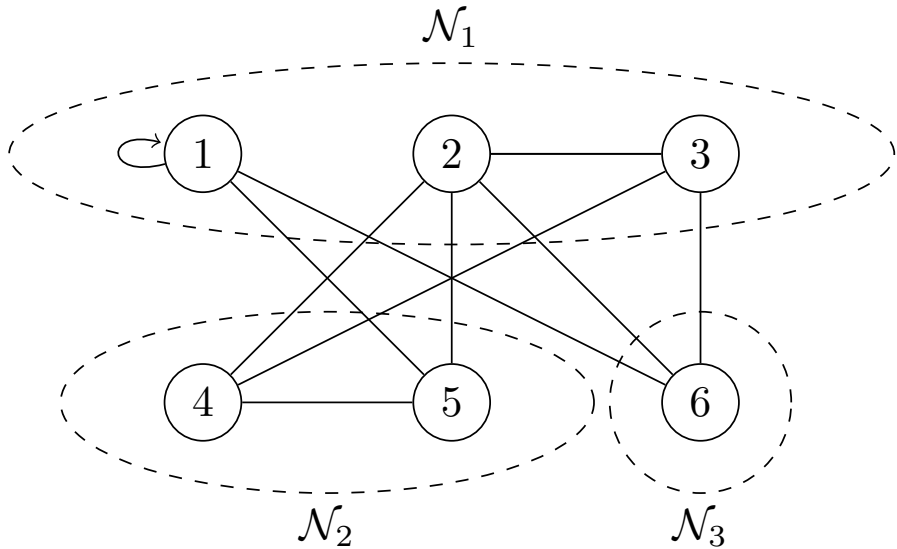}
    \caption{\textbf{Graph with a partition of the node set.} A graph with $N=6$ nodes and the partition $\pi=\{\mathcal{N}_1,\mathcal{N}_2,\mathcal{N}_3 \}$, whose cells are given by $\mathcal{N}_1 = \{1, 2, 3\}$, $\mathcal{N}_2=\{4, 5\}$ and $\mathcal{N}_3=\{6\}$. For unit link weights, i.e., $\beta_{ij}=1$ for all nodes $i,j$, the partition $\pi$ is \textit{not} equitable. If the link weights $\beta_{ij}$ satisfy (\ref{ljnagergergaaaa}), as in Example~\ref{example:equitable_partition}, then the partition $\pi$ is equitable. \label{fig:equitable_partition}}
\end{figure}    
Consider NIMFA on a graph with $N=6$ nodes, whose curing rate matrix equals $S = \operatorname{diag}\left( \tilde{\delta}_1, \tilde{\delta}_1,\tilde{\delta}_1,\tilde{\delta}_2,\tilde{\delta}_2, \tilde{\delta}_3\right)$ for some curing rates $\tilde{\delta}_1,\tilde{\delta}_2, \tilde{\delta}_3$. Furthermore, suppose that the infection rate matrix~$B$ is symmetric and given by the graph in Figure~\ref{fig:equitable_partition} as
\begin{equation}
B = \left(\begin{NiceArray}{ccc:cc:c}
\beta_{11} & 0 & 0 & 0 & \beta_{15} & \beta_{16} \\
0 & 0 & \beta_{23} & \beta_{24} & \beta_{25}& \beta_{26} \\
0 & \beta_{23} & 0 & \beta_{34}  & 0 &	\beta_{36} \\
\hdottedline
0 & \beta_{24} & \beta_{43} & 0 & \beta_{45} &0\\
\beta_{15} & \beta_{25} & 0 & \beta_{45} & 0 &	0  \\
\hdottedline
\beta_{16} & \beta_{26} & \beta_{36} & 0 & 0 & 0
\end{NiceArray}\right).
\end{equation}
Suppose that, for some degrees $d_{pl}>0$, the infection rates $\beta_{ij}$ satisfy: $\beta_{11}=\beta_{23}=d_{11}$; $\beta_{15}=\beta_{34}=d_{12}$ and $\beta_{24}=\beta_{25}=d_{12}/2$; $\beta_{16}=\beta_{26}=\beta_{36}=d_{13}$; and $\beta_{45}=d_{22}$. Then, the infection rate matrix $B$ becomes
\begin{equation}\label{ljnagergergaaaa}
B = \left(\begin{NiceArray}{ccc:cc:c}
d_{11} & 0 & 0 & 0 & d_{12} & d_{13} \\
0 & 0 & d_{11} & d_{12}/2 & d_{12}/2 & d_{13} \\
0 & d_{11} & 0 & d_{12}  & 0 &	d_{13}  \\
\hdottedline
0 & d_{12}/2 & d_{12} & 0 & d_{22} &0\\
d_{12} & d_{12}/2 & 0 & d_{22} & 0 &	0  \\
\hdottedline
d_{13} & d_{13} & d_{13} & 0 & 0 & 0
\end{NiceArray}\right).
\end{equation}
Thus, the matrix $B$ has the equitable partition $\pi=\{\mathcal{N}_1,\mathcal{N}_2,\mathcal{N}_3 \}$ with the cells $\mathcal{N}_1 = \{1, 2, 3\}$, $\mathcal{N}_2=\{4, 5\}$ and $\mathcal{N}_3=\{6\}$. The quotient matrix equals 
\begin{align}
B^{\pi} = \begin{pmatrix}
d_{11} & d_{12}  & d_{13} \\
d_{12} & d_{22}  & 0 \\
d_{13} & 0  & 0 
\end{pmatrix}.
\end{align}
 For the partition $\pi$, the reduced-size viral state can be chosen\footnote{But, for instance, $v^{\pi}(t) = \left(v_{2}(t), v_{5}(t), v_{6}(t)\right)^T$ is possible as well.} as $v^{\pi}(t) = \left(v_{1}(t), v_{4}(t), v_{6}(t)\right)^T$. Theorem~\ref{theorem:equitableOriginal} states that the vector $v^{\pi}(t) = \left(v_{1}(t), v_{4}(t), v_{6}(t)\right)^T$ evolves as
\begin{align}
	\frac{d v^{\pi}(t)}{d t } & = - S^{\pi}  v^{\pi}(t) + \textup{\textrm{diag}}\left( u_3 - v^{\pi}(t)\right)B^{\pi}  v^{\pi}(t),
\end{align}
with the $3\times 3$ reduced-size curing rate matrix $S^{\pi} = \operatorname{diag}\left( \tilde{\delta}_1,\tilde{\delta}_2, \tilde{\delta}_3 \right)$. Furthermore, Corollary~\ref{corollary:equitable_are_invariant} states that the viral state $v(t)$ has the proper orthogonal decomposition
\begin{align}
v(t) = \sqrt{3} v^{\pi}_1(t) y_1 + \sqrt{2} v^{\pi}_2(t)y_2 + v^{\pi}_3(t)y_3
\end{align}
with the agitation modes
\begin{align}
y_1 &= \frac{1}{\sqrt{3}} \begin{pmatrix}
1 & 1& 1& 0 & 0 & 0
\end{pmatrix}^T,\\
y_2 &= \frac{1}{\sqrt{2}} \begin{pmatrix}
0 & 0& 0& 1 & 1 & 0
\end{pmatrix}^T,\\
y_3 &= \begin{pmatrix}
0 & 0& 0& 0 & 0 & 1
\end{pmatrix}^T.
\end{align}
\end{example}

\section{Exact clustering}

Theorem~\ref{theorem:equitableOriginal} and Corollary~\ref{corollary:equitable_are_invariant} only give an incomplete answer to Problem~\ref{problem:clustering}: if the infection rate matrix $B$ has an equitable partition $\pi$, then there exists an invariant set $\mathcal{V}$. But are there invariant sets $\mathcal{V}$, even if the matrix $B$ does not have an equitable partition $\pi$?

We denote the \textit{orthogonal complement} of the viral state set $\mathcal{V}$ by
\begin{align}\label{def:orthogonal_viral_space_V}
\mathcal{V}^\bot = \left\{ w \in \mathbb{R}^N | w^T v = 0, \quad \forall v \in \mathcal{V}\right\}.
\end{align}
The dimension of the set $\mathcal{V}$ equals $m$. Thus, the dimension of the orthogonal complement $\mathcal{V}^\bot$ equals $N-m$. Since the orthogonal complement $\mathcal{V}^\bot$ is a subspace, there is a set of $N-m$ orthonormal basis vectors $y_{m+1}, ..., y_N$ such that
\begin{align} \label{V_bot_def}
\mathcal{V}^\bot=\operatorname{span}\{y_{m+1}, ..., y_N\}.
\end{align}
The \textit{direct sum} of two subspaces $\mathcal{S}_1, \mathcal{S}_2\subseteq \mathbb{R}^N$ is defined as the subspace 
\begin{align}\label{def:direct_sum}
\mathcal{S}_1\oplus \mathcal{S}_2 = \left\{s_1 + s_2 | s_1 \in \mathcal{S}_1, s_2 \in \mathcal{S}_2 \right\}.
\end{align}
Thus, the Euclidean space is the direct sum $\mathbb{R}^N = \textrm{span}\{\mathcal{V}\} \oplus \mathcal{V}^\bot$ of the two subspaces~$\mathcal{V}, \mathcal{V}^\bot$.

We rely on four assumptions to solve Problem~\ref{problem:clustering}. 
\begin{assumption}\label{assumption:delta_v_in_V}
For every viral state $v\in \mathcal{V}$, \textit{we require} that $\mathrm{diag}\left( \delta_1, ..., \delta_N \right) v \in \mathcal{V}$.
\end{assumption}
Suppose that the curing rates are homogeneous, i.e., $\delta_i=\delta$ for all nodes $i$. Then, Assumption~\ref{assumption:delta_v_in_V} is satisfied, since $\mathrm{diag}\left( \delta_1, ..., \delta_N \right) v = \delta v \in \mathcal{V}$ for every viral state $v\in\mathcal{V}$. More generally, Assumption~\ref{assumption:delta_v_in_V} states that the viral state set $\mathcal{V}$ is an invariant subspace of the curing rate matrix $\mathrm{diag}\left( \delta_1, ..., \delta_N \right)$. Intuitively speaking, the curing rates $\delta_1, ..., \delta_N$ are \quotes{set in accordance to} the clustering given by the viral state set $\mathcal{V}$, such as in Example~\ref{example:equitable_partition}. 

\begin{assumption}\label{assumption:v_positive}
There is a viral state $v\in\mathcal{V}$ whose entries satisfy $v_i>0$ for every node $i=1, ..., N$.
\end{assumption}
If $R_0>1$ and the matrix $B$ is irreducible, then \cite{lajmanovich1976deterministic} there is a unique steady-state $v_\infty$ with positive components $v_{\infty,i}>0$. Since every viral state $v$ converges to the steady state $v_\infty$, the steady state $v_\infty$ is element of the invariant set $\mathcal{V}$. Hence, Assumption~\ref{assumption:v_positive} is always satisfied if $R_0>1$, provided the matrix $B$ is irreducible.

\begin{assumption}\label{assumption:spreading_rates}
The curing rates are positive and the infection rates are non-negative, i.e., $\delta_i>0$ and $\beta_{ij}\ge 0$ for all nodes $i,j$.
\end{assumption}

Assumption~\ref{assumption:spreading_rates} is rather technical, since only non-negative curing rates and infection rates have a physical meaning.

\begin{assumption}\label{assumption:matrix_B_symmetric}
The infection rate matrix~$B$ is symmetric and irreducible.
\end{assumption}
Assumption~\ref{assumption:matrix_B_symmetric} holds if and only if the infection rate matrix $B$ corresponds to a connected undirected graph \cite{van2014performance}. Under Assumption~\ref{assumption:matrix_B_symmetric}, the matrix~$B$ is diagonalisable \cite{van2010graph} as
\begin{align} \label{diag_B}
B = X \Lambda X^T.
\end{align}
Here, we denote the $N\times N$ diagonal matrix $\Lambda=\operatorname{diag}(\lambda_1, ..., \lambda_N)$ whose diagonal entries are given by the real eigenvalues $\lambda_1\ge \lambda_2 \ge ...\ge \lambda_N$, and the columns of the $N\times N$ matrix $X=(x_1, ..., x_N)$ are given by the corresponding eigenvectors $x_i$. 

Lemma~\ref{lemma:invariance} states that the invariant set $\mathcal{V}$ and the orthogonal complement $\mathcal{V}^\bot$ are spanned by eigenvectors of the infection rate matrix $B$:

\begin{lemma} \label{lemma:invariance}
Suppose that Assumptions~\ref{assumption:delta_v_in_V} and \ref{assumption:matrix_B_symmetric} hold, and consider an invariant set $\mathcal{V}=\operatorname{span}\{y_1, ..., y_m\}$ of NIMFA (\ref{NIMFA_stacked}) and the orthogonal complement~$\mathcal{V}^\bot=\operatorname{span}\{y_{m+1}, ..., y_N\}$. Then, there is some permutation~$\phi:\{1, ..., N\} \rightarrow \{1, ..., N\}$, such that $\mathcal{V}=\operatorname{span}\{x_{\phi(1)}, ..., x_{\phi(m)}\}$ and $\mathcal{V}^\bot=\operatorname{span}\{x_{\phi(m+1)}, ..., x_{\phi(N)}\}$, where $x_{\phi(1)}, ..., x_{\phi(N)}$ denotes an orthonormal set of eigenvectors of the infection rate matrix~$B$ to the eigenvalues~$\lambda_{\phi(1)}, ..., \lambda_{\phi(N)}$.
\end{lemma}
\begin{proof}
Appendix \ref{appendix:lemma_invariance}
\end{proof}
 We denote the span of the vectors~$x_{\phi(l)}$ of the subspace~$\mathcal{V}$ which correspond to a \textit{non-zero} eigenvalue~$\lambda_{\phi(l)}\neq 0$ as $\mathcal{V}_{\neq 0} = \operatorname{span}\left\{ x_{\phi(l)} \big| l=1, ..., m, \lambda_{\phi(l)}\neq 0 \right\}$. Let the number of non-zero eigenvalues be denoted by $m_1$. Without loss of generality, we assume that, after the permutation $\phi$, the \textit{first} $m_1$ eigenvalues~$\lambda_{\phi(1)}, ..., \lambda_{\phi(m_1)}$ are non-zero. Hence, the subspace $\mathcal{V}_{\neq 0}$ equals 
\begin{align}\label{V_0_def}
\mathcal{V}_{\neq 0} = \operatorname{span}\left\{ x_{\phi(l)} \big| l=1, ..., m_1 \right\}.
\end{align}
Analogously to (\ref{V_0_def}), we define the span of the vectors~$x_{\phi(l)}$ of the subspace~$\mathcal{V}$ which correspond to a \textit{zero} eigenvalue~$\lambda_{\phi(l)}= 0$ as
\begin{align}
\mathcal{V}_0 &= \operatorname{span}\left\{ x_{\phi(l)} \big| l=1, ..., m, \lambda_{\phi(l)}= 0 \right\} \\
&=\operatorname{span}\left\{ x_{\phi(l)} \big| l=m_1+1, ..., m \right\}.
\end{align}
Thus, the subspace $\mathcal{V}$ is equal to the direct sum
\begin{align}\label{sdfasgrfasfasdfasdf}
\mathcal{V} = \mathcal{V}_{\neq 0} \oplus \mathcal{V}_0. 
\end{align}
We emphasise that $\operatorname{span}\left\{ y_1, ..., y_{m} \right\}=\operatorname{span}\left\{ x_{\phi(1)}, ..., x_{\phi(m)} \right\}$ does not imply that $y_l=x_{\phi(k)}$ for some $k,l$. An immediate consequence of Lemma~\ref{lemma:invariance} is that the infection rate matrix $B$ can be decomposed as:
\begin{lemma}\label{lemma:B_decomposition}
Suppose that Assumptions~\ref{assumption:delta_v_in_V} and \ref{assumption:matrix_B_symmetric} hold, and consider an invariant set $\mathcal{V}=\operatorname{span}\{y_1, ..., y_m\}$ of NIMFA (\ref{NIMFA_stacked}) and the orthogonal complement~$\mathcal{V}^\bot=\operatorname{span}\{y_{m+1}, ..., y_N\}$. Then, the infection rate matrix $B$ is decomposable as $B=B_{\mathcal{V}}+B_{\mathcal{V}^\bot}$, where 
\begin{align}
B_{\mathcal{V}} = \begin{pmatrix}
y_1 & ... & y_{m}
\end{pmatrix}
\tilde{B}_{\mathcal{V}}
 \begin{pmatrix}
y^T_1 \\
\vdots \\
 y^T_{m}
\end{pmatrix} \quad \text{and} \quad B_{\mathcal{V}^\bot} = \begin{pmatrix}
y_{m+1} & ... & y_{N}
\end{pmatrix}
\tilde{B}_{\mathcal{V}^\bot}
 \begin{pmatrix}
y^T_{m+1} \\
\vdots \\
 y^T_{N}
\end{pmatrix}
\end{align}
for some $m\times m$ matrix $\tilde{B}_{\mathcal{V}}$ and $(N-m)\times (N-m)$ matrix $\tilde{B}_{\mathcal{V}^\bot}$.
\end{lemma}
\begin{proof}
Appendix~\ref{appendix:B_decomposition}.
\end{proof}
Lemma \ref{lemma:B_decomposition} shows that the sets $\mathcal{V}$ and $\mathcal{V}^\bot$ are invariant subspaces of the matrix $B$. In particular, the viral state dynamics on the invariant set $\mathcal{V}$ are the same for all infection rate matrices $B^{(1)},B^{(2)}$ with the same submatrix $B^{(1)}_{\mathcal{V}}=B^{(2)}_{\mathcal{V}}$ but different submatrices $B^{(1)}_{\mathcal{V}^\bot}\neq B^{(2)}_{\mathcal{V}^\bot}$.

\begin{example}
Suppose that Assumptions~\ref{assumption:delta_v_in_V} and \ref{assumption:matrix_B_symmetric} hold. For some degrees $d_{11}, d_{12}, d_{22}$ and some scalar $\xi$, consider the infection rate matrix
\begin{equation}\label{ljnagergerg}
B = \left(\begin{NiceArray}{cc:c}
d_{11} +\xi & d_{11}-\xi & d_{12}  \\
d_{11} -\xi & d_{11}+\xi & d_{12}  \\
\hdottedline
d_{12} & d_{12} & d_{22}  
\end{NiceArray}\right)
\end{equation}
with the equitable partition $\pi=\left\{\mathcal{N}_1, \mathcal{N}_2\right\}$, where $\mathcal{N}_1=\{1, 2\}$ and $\mathcal{N}_2=\{3\}$, and the quotient matrix
\begin{align}
B^{\pi}= \begin{pmatrix}
d_{11} & d_{12} \\
d_{12} & d_{22}
\end{pmatrix}.
\end{align}
Corollary~\ref{corollary:equitable_are_invariant} states that the subspace $\mathcal{V}=\operatorname{span}\{y_1, y_2\}$ is an invariant set of NIMFA (\ref{NIMFA_stacked}), where the agitation modes are equal to $y_1=\frac{1}{\sqrt{2}}(1, 1, 0)^T$ and $y_2 = (0, 0,1)^T$. The orthogonal complement follows as $\mathcal{V}^\bot=\operatorname{span}\{y_3\}$, where $y_3=\frac{1}{\sqrt{2}}(1, -1, 0)^T$. Furthermore, Lemma~\ref{lemma:B_decomposition} states that the infection rate matrix can be decomposed as $B=B_{\mathcal{V}}+B_{\mathcal{V}^\bot}$, where
\begin{align}
B_{\mathcal{V}} &= 
\begin{pmatrix}
y_1 & y_2
\end{pmatrix}
\begin{pmatrix}
2 d_{11} & \sqrt{2}d_{12} \\
\sqrt{2}d_{12} & d_{22} \\
\end{pmatrix}
\begin{pmatrix}
y^T_1 \\
y^T_2
\end{pmatrix}= \begin{pmatrix}
d_{11} & d_{11} & d_{12}  \\
d_{11} & d_{11}& d_{12}  \\
d_{12} & d_{12} & d_{22}  
\end{pmatrix}
\end{align}
and
\begin{align}
B_{\mathcal{V}^\bot} &= 2\xi y_3  y^T_3 = \begin{pmatrix}
\xi & -\xi & 0  \\
-\xi & \xi & 0  \\
0 & 0 & 0
\end{pmatrix}.
\end{align}
The eigenvectors $x_{\phi(1)}$, $x_{\phi(2)}$ are equal to a linear combination of the agitation modes $y_1$, $y_2$, and the third eigenvector equals $x_{\phi(3)}=y_3$.
\end{example}
Theorem~\ref{theorem:invariant_sets_are_equitable_partitions} states our main result:  
\begin{theorem}\label{theorem:invariant_sets_are_equitable_partitions}
Suppose that Assumptions~\ref{assumption:delta_v_in_V} to \ref{assumption:matrix_B_symmetric} hold. Then, any invariant set $\mathcal{V}=\operatorname{span}\left\{ y_1, ..., y_m\right\}$ of NIMFA (\ref{NIMFA_stacked}) is equal to the direct sum $\mathcal{V}=\mathcal{V}_{\neq 0} \oplus \mathcal{V}_0$ of two subspaces $\mathcal{V}_{\neq 0}$,$\mathcal{V}_0$. Here, the orthonormal basis vectors $y_1, ..., y_{m_1}$, where $m_1\le m$, of the subspace $\mathcal{V}_{\neq 0} =\operatorname{span}\left\{y_1, ..., y_{m_1}\right\}$ are given by
\begin{align}\label{y_l_definition}
(y_l)_i = \begin{cases}
\frac{1}{\sqrt{|\mathcal{N}_l|}}\quad &\text{if}~i\in\mathcal{N}_l, \\
0 &\text{if}~i\not\in\mathcal{N}_l,
\end{cases}
\end{align}  
for some \emph{equitable} partition $\pi=\left\{\mathcal{N}_1, ..., \mathcal{N}_{m_1}\right\}$ of the infection rate matrix $B$. If $m_1=m$, then the subspace $\mathcal{V}_0$ is empty. Otherwise, if $m_1<m$, then $\mathcal{V}_0=\operatorname{span}\left\{ x_{\phi(l)} \big| l=m_1+1, ..., m \right\}$ for some eigenvectors $x_{\phi(l)}$ of the infection rate matrix $B$ belonging to the eigenvalue 0.
\end{theorem}
\begin{proof}
Appendix \ref{appendix:invariant_sets_are_equitable_partitions}.
\end{proof}

The Euclidean space $\mathbb{R}^N$ is always an invariant set of NIMFA. For $\mathcal{V}=\mathbb{R}^N$ and $\mathcal{V}_0=\emptyset$, the equitable partition $\pi$ in Theorem~\ref{theorem:invariant_sets_are_equitable_partitions} becomes \textit{trivial}, i.e., $\pi=\left\{\mathcal{N}_1, ..., \mathcal{N}_{N}\right\}$ with exactly one node in every cell $\mathcal{N}_l$. On the other hand, if there is an invariant set $\mathcal{V}$ of dimension $m<N$, then Theorem~\ref{theorem:invariant_sets_are_equitable_partitions} implies that the matrix $B$ is equitable with $m_1\le m$ cells.

If $\mathcal{V}_0=\emptyset$, then Theorem~\ref{theorem:invariant_sets_are_equitable_partitions} essentially reverts Corollary~\ref{corollary:equitable_are_invariant}. Thus, every equitable partition $\pi$ corresponds to an invariant set $\mathcal{V}_0$, and vice versa. \textit{In other words, the macroscopic structure of equitable partitions $\pi$ and the low-rank dynamics of invariant sets $\mathcal{V}$ are two sides of the same coin.} If $\mathcal{V}_0=\emptyset$, then the dynamics on the invariant set $\mathcal{V}=\mathcal{V}_{\neq 0}$ are given by the reduced-size NIMFA system (\ref{kjnkjnertgverg}) with $m=m_1$ equations.

If $\mathcal{V}_0\neq \emptyset$, then Theorem~\ref{theorem:invariant_sets_are_equitable_partitions} is more general than the inversion of Corollary~\ref{corollary:equitable_are_invariant}. Theorem~\ref{theorem:invariant_sets_are_equitable_partitions} states that invariant set of NIMFA is equal to the direct sum $\mathcal{V}=\mathcal{V}_{\neq 0}\oplus\mathcal{V}_0$, where the subspace $\mathcal{V}_{\neq 0}$ corresponds to an equitable partition $\pi$ of the infection rate matrix, and the subspace $\mathcal{V}_0$ is a subset of the kernel of the matrix $B$. If $\mathcal{V}_0\neq \emptyset$, then the dynamics on the invariant set $\mathcal{V}=\mathcal{V}_{\neq 0}\oplus\mathcal{V}_0$ are described by the $m>m_1$ differential equations (\ref{aergtgrtedgtasgtsrht}). 

The curing rates $\delta_i$ satisfy Assumption~\ref{assumption:delta_v_in_V} if there are some scalars $\tilde{\delta}_1, ..., \tilde{\delta}_{m_1}$ such that $\delta_i= \tilde{\delta}_l$ for all nodes $i$ in cell $\mathcal{N}_l$, where $l=1, ..., m_1$. However, Assumption~\ref{assumption:delta_v_in_V} allows for more general curing rates. With Lemma~\ref{lemma:B_decomposition} and Theorem~\ref{theorem:invariant_sets_are_equitable_partitions}, the infection rate matrix $B$ can be constructed from specifying the agitation modes $y_l$, such that $\mathcal{V}=\operatorname{span}\{y_1, ..., y_m\}$ is an invariant set of NIMFA (\ref{NIMFA_stacked}):
\begin{example}\label{example:construct_B_from_y}
Consider NIMFA (\ref{NIMFA_stacked}) on a network of $N=5$ nodes and the subspaces $\mathcal{V}_{\neq 0}=\operatorname{span}\{y_1,y_2\}$, $\mathcal{V}_0 =\operatorname{span}\{y_3\}$, where the agitation modes equal
\begin{align}
y_1&= \frac{1}{\sqrt{3}} \begin{pmatrix}
1 & 1 & 1 & 0 & 0
\end{pmatrix}^T, \\
y_2&= \frac{1}{\sqrt{2}} \begin{pmatrix}
0 & 0 & 0 & 1 & 1
\end{pmatrix}^T, \\
y_3&= \frac{1}{\sqrt{6}} \begin{pmatrix}
1 & -2 & 1 & 0 & 0
\end{pmatrix}^T. 
\end{align}
Furthermore, let $y_4, y_5$ be two vectors, with $y^T_4 y_5=0$ and $y^T_4 y_4=y^T_5 y_5=1$, that are orthogonal to the agitation modes $y_1, y_2, y_3$. With Lemma~\ref{lemma:B_decomposition}, define the infection rate matrix as
\begin{align}
B = \begin{pmatrix}
y_1 & y_2
\end{pmatrix}
\tilde{B}_{\mathcal{V}_{\neq 0}}
\begin{pmatrix}
y^T_1 \\
y^T_2
\end{pmatrix}
+\begin{pmatrix}
y_4 & y_5
\end{pmatrix}
\tilde{B}_{\mathcal{V}^\bot}
\begin{pmatrix}
y^T_4 \\
y^T_5
\end{pmatrix},
\end{align}
where the symmetric $2\times 2$ matrices $\tilde{B}_{\mathcal{V}_{\neq 0}}$, $\tilde{B}_{\mathcal{V}^\bot}$ are chosen such that the matrix $B$ is irreducible and contains only non-negative elements. Furthermore, consider the curing rate matrix $S = \operatorname{diag}(\tilde{\delta}_1, \tilde{\delta}_2, \tilde{\delta}_1, \tilde{\delta}_3, \tilde{\delta}_3  )$ for some curing rates $\tilde{\delta}_1, \tilde{\delta}_2, \tilde{\delta}_3>0$. Then, Assumptions~\ref{assumption:delta_v_in_V} to \ref{assumption:matrix_B_symmetric} are satisfied, and Theorem~\ref{theorem:invariant_sets_are_equitable_partitions} states that the subspace $\mathcal{V}=\mathcal{V}_{\neq 0}\oplus\mathcal{V}_0$ is an invariant set of NIMFA (\ref{NIMFA_stacked}). (An alternative choice for the curing rate matrix is $S = \operatorname{diag}(\tilde{\delta}_1, \tilde{\delta}_1, \tilde{\delta}_1, \tilde{\delta}_2, \tilde{\delta}_2 )$, which also satisfies Assumption~\ref{assumption:delta_v_in_V}.)
\end{example}

In \cite{prasse2019time}, we derived the solution of the NIMFA model (\ref{NIMFA_stacked}) around the epidemic threshold $R_0=1$. More precisely, under mild assumptions, we derived the approximation $v_\text{apx}(t)=c(t)v_\infty$ with an explicit, closed-form expression for the scalar function $c(t)$. If the initial viral state satisfies $\lVert v(0)\rVert_2\le \tilde{\sigma} (R_0-1)^2$ for some constant $\tilde{\sigma}$ as $R_0\downarrow 1$, then it holds that $\lVert v(t) - v_\text{apx}(t)\rVert_2 \le \sigma (R_0-1)^2$ at every time $t$ for some constant $\sigma$ as $R_0\downarrow 1$. Hence, the viral state $v(t)$ converges to the approximation $v_\text{apx}(t)$ uniformly in time $t$. Remarkably, since $v_\text{apx}=c(t)v_\infty$, the viral state $v(t)$ lies in the one-dimensional subspace $\mathcal{V}=\operatorname{span}\{v_\infty\}$ when $R_0\downarrow 1$, for an arbitrarily large and heterogeneous contact network. Figure~\ref{fig:sfwrfetgetg} illustrates the uniform convergence result in \cite[Theorem 3]{prasse2019time}.

\begin{figure}[!ht]
    \centering 
\includegraphics[width=0.6\textwidth]{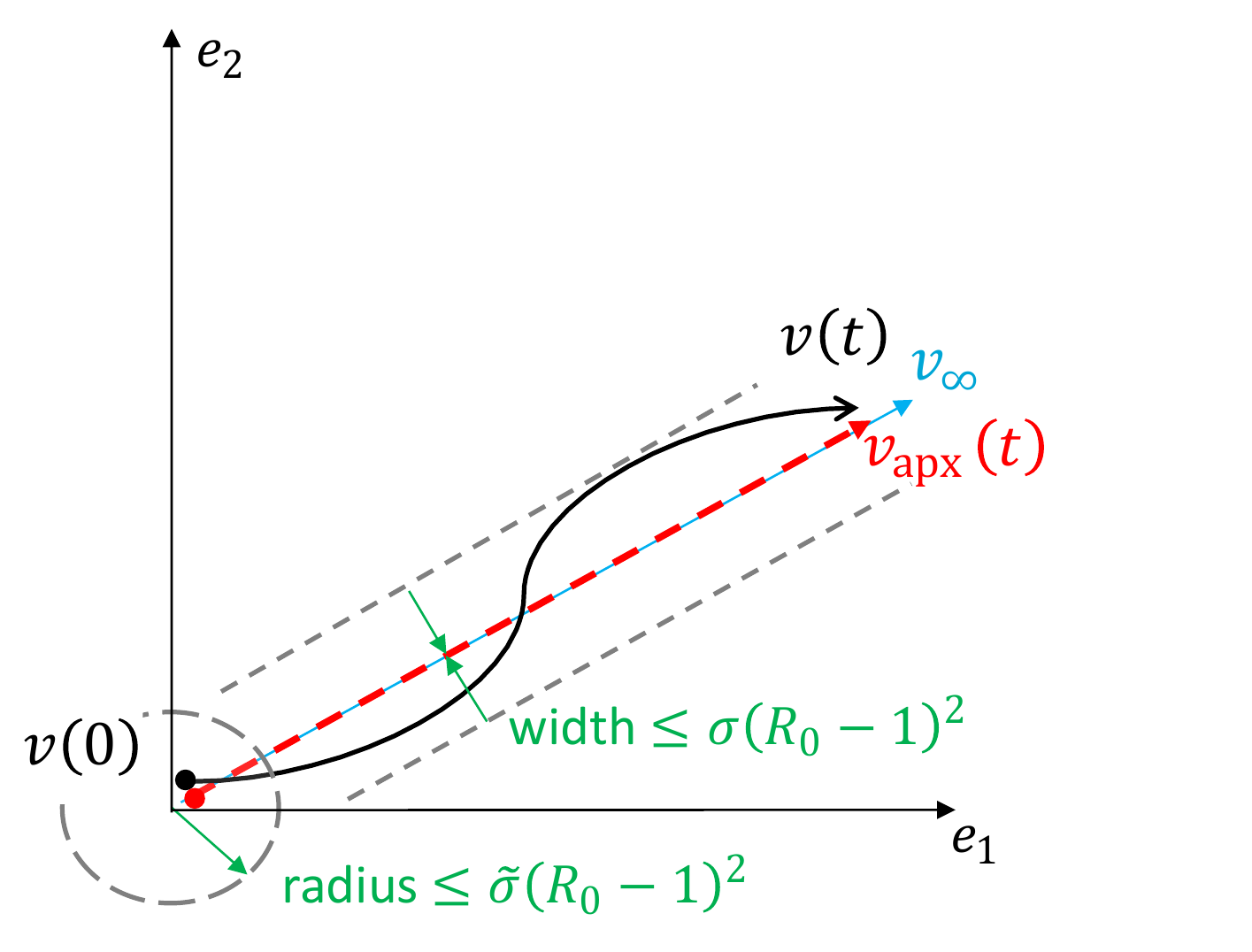} 
\caption{\textbf{Viral dynamics around the epidemic threshold $R_0=1$.} An illustration of the uniform convergence result in \cite[Theorem 3]{prasse2019time} for a network with $N=2$ nodes. The black curve shows the trajectory of the $2\times 1$ viral state vector $v(t)$ as time $t$ evolves. The blue line shows the steady state $v_\infty$. The red curve depicts the trajectory closed-form approximation $v_\text{apx}(t)=c(t)v_\infty$, which is in the subspace $\operatorname{span}\{v_\infty\}$ at every time $t$. If the initial viral state $v(0)$ is positive and in the disk of radius $\tilde{\sigma} \left( R_0 -1 \right)^2$ for some constant $\tilde{\sigma}$, then the approximation error $\lVert v(t) - v_\text{apx}(t)\rVert_2$ is bounded by $\sigma \left( R_0 -1 \right)^2$ for some constant $\sigma$ at every time $t$ as $R_0\downarrow 1$.}\label{fig:sfwrfetgetg}
\end{figure}

As illustrated by Figure~\ref{fig:sfwrfetgetg}, the viral state $v(t)$ converges to the one-dimensional dynamics $v_\text{apx}(t)$ as $R_0\downarrow 1$. \textit{Are there networks for which the approximation $v_\textup{apx}(t)$ is exact, for any basic reproduction number $R_0>1$?} The infection rate matrix $B$ is \textit{regular} if 
\begin{align} \label{asdafsdfseddasss}
\sum^N_{k=1} \beta_{ik} =\sum^N_{k=1} \beta_{jk}
\end{align}
for all nodes $i,j$. From Theorem~\ref{theorem:invariant_sets_are_equitable_partitions}, we obtain:

\begin{corollary}\label{corollary:one_dimensional_solution}
Suppose that Assumptions~\ref{assumption:delta_v_in_V} to \ref{assumption:matrix_B_symmetric} hold and consider that $R_0>1$. Then, there is an $m=1$ dimensional invariant set $\mathcal{V}=\operatorname{span}\{y_1\}$ of NIMFA (\ref{NIMFA_stacked}) if and only if $\mathcal{V}_0=\emptyset$, the agitation mode equals either $y_1=v_\infty/\lVert v_\infty\rVert_2$ or $y_1=-v_\infty/\lVert v_\infty\rVert_2$ and the infection rate matrix $B$ is regular. Furthermore, the approximation $v_\textup{apx}(t)=c(t)v_\infty$ is exact if and only if the matrix $B$ is regular and $v(0)=c(0)v_\infty$ for some scalar $c(0)$.
\end{corollary}
\begin{proof}
Appendix~\ref{appendix:one_dimensional_solution}
\end{proof}

\subsection{Decomposition of the viral dynamics}

Suppose the infection rate matrix $B$ has an equitable partition $\pi$ and the infection rates $\beta_{ij}$ are the same between all nodes $i,j$ in any two cells\footnote{If the matrix $B$ is decomposable as $B=B_{\mathcal{V}}+B_{\mathcal{V}^\bot}$ as in Lemma~\ref{lemma:B_decomposition}, then the infection rates $\beta_{ij}$ are the same between all nodes $i,j$ in any two cells if and only if $B_{\mathcal{V}^\bot}=0$.}. Then, we can decompose the dynamics of the viral state $v(t)$ as:
\begin{theorem}\label{theorem:decompose_dynamics}
Consider NIMFA (\ref{NIMFA_stacked}) on a symmetric $N\times N$ infection rate matrix $B$ with an equitable partition $\pi= \{\mathcal{N}_1, ..., \mathcal{N}_r\}$. Furthermore, suppose that the curing rates $\delta_i$ are the same for all nodes $i$ in any cell $\mathcal{N}_l$, and that the infection rates $\beta_{ij}$ are the same for all nodes $i$ in any cell $\mathcal{N}_l$ and all nodes $j$ in any cell $\mathcal{N}_p$. Denote the subspace $\mathcal{V}_{\neq 0}=\operatorname{span}\{y_1, ..., y_r\}$, with the basis vectors $y_l$ defined in (\ref{y_l_definition}), and denote the kernel of the matrix $B$ by $\operatorname{ker}(B)=\operatorname{span}\{y_{r+1}, ..., y_N\}$. At every time $t\ge 0$, consider the viral state decomposition 
\begin{align}\label{khbjkhbjkrergregrgreg}
v(t) = \tilde{v}(t) + v_\textup{ker}(t),
\end{align}
where the projection of the viral state $v(t)$ on the subspace $\mathcal{V}_{\neq 0}$ equals
\begin{align}
\tilde{v}(t) = \sum^{r}_{l=1} \left( y^T_l v(t) \right) y_l,
\end{align}
and the projection of the viral state $v(t)$ on the kernel $\operatorname{ker}(B)$ equals
\begin{align}
v_\textup{ker}(t) = \sum^{N}_{l=r+1} \left( y^T_l v(t) \right) y_l.
\end{align}
Furthermore, denote the $r\times 1$ reduced-size projection $\tilde{v}^{\pi}(t)= \left(\tilde{v}^{\pi}_{i_1}(t), ..., \tilde{v}^{\pi}_{i_r}(t)\right)^T$, where $i_l$ denotes an arbitrary node in cell $\mathcal{N}_l$. Then, the reduced-size projection $\tilde{v}^{\pi}(t)$ evolves, independently of the projection $v_\textup{ker}(t)$, as
\begin{align}\label{kjnksdfsdfsdf}
	\frac{d \tilde{v}^{\pi}(t)}{d t } & = - S^{\pi} \tilde{v}^{\pi}(t) + \textup{\textrm{diag}}\left( u_r - \tilde{v}^{\pi}(t)\right) B^{\pi}  \tilde{v}^{\pi}(t)
\end{align}
with the quotient matrix $B^{\pi}$ and the matrix $S^{\pi}$ given by (\ref{definition_S_pi}), and the projection $v_\textup{ker}(t)$ obeys
\begin{align}\label{sadfsdgadfdddaaa}
\frac{d v_\textup{ker}(t)}{dt} = - \left( S+ \operatorname{diag}\left( B \tilde{v}(t)\right) \right)v_\textup{ker}(t).
\end{align}
\end{theorem}
\begin{proof}
Appendix~\ref{appendix:decompose_dynamics}.
\end{proof}

In Theorem~\ref{theorem:decompose_dynamics}, the set $\mathcal{V}_0$ is equal to the kernel $\operatorname{ker}(B)$, which is equivalent to $\mathcal{V}^\bot=\emptyset$ and assuming the same infection rates $\beta_{ij}$ between all nodes $i,j$ in any two cells. In contrast to Theorem~\ref{theorem:equitableOriginal}, we do not consider that the initial state satisfies $v_i(0)=v_j(0)$ for all nodes $i,j$ in the same cell $\mathcal{N}_l$. 

With the definition of the agitation mode $y_l$ in (\ref{y_l_definition}), the viral state average in cell $\mathcal{N}_l$ follows from the projection of the viral state $v(t)$ on the vector $y_l$ as
\begin{align}
\frac{1}{\left|\mathcal{N}_l \right|}\sum_{i \in \mathcal{N}_l} v_i(t) = \frac{1}{\sqrt{\left|\mathcal{N}_l \right|}} y^T_l v(t)
\end{align}
for every cell $l=1, ..., r$. Furthermore, the subspace $\mathcal{V}_{\neq 0}$ is spanned by the vectors $y_1$, ..., $y_r$. Hence, the dynamics of the projection $\tilde{v}(t)$ on the subspace $\mathcal{V}_{\neq 0}$ describes the evolution of viral state averages of every cell $\mathcal{N}_l$, which is described by $r$ differential equations (\ref{kjnksdfsdfsdf}) on the quotient graph $B^\pi$. Since the steady state $v_{\infty, i}$ of every node $i$ in the same cell $\mathcal{N}_l$ is the same \cite{bonaccorsi2015epidemic, ottaviano2018optimal}, it holds that $v_\infty\in \mathcal{V}_{\neq 0}$, which implies that $v_\textup{ker}(t)\rightarrow 0$ as $t\rightarrow \infty$. Furthermore, from Theorem~\ref{theorem:equitableOriginal} it follows that, if $v_\textup{ker}(0)=0$, then $v_\textup{ker}(t)=0$ at every time $t$. Thus, the evolution of the projection $v_\textup{ker}(t)$ describes convergence of the viral states $v_i(t)$ to the respective cell-averages. \textit{By (\ref{kjnksdfsdfsdf}), Theorem~\ref{theorem:decompose_dynamics} implies that the viral state cell-averages evolve independently of the dynamics on the kernel $\operatorname{ker}(B)$.} Schaub \textit{et al.} \cite{schaub2016graph} obtained an analogous result for linear dynamics on networks.
 
If we can derive the closed-form expression for the projection $\tilde{v}(t)$ by solving (\ref{kjnksdfsdfsdf}), then the dynamics $v_\textup{ker}(t)$ follow by the linear time-\textit{varying} system (\ref{sadfsdgadfdddaaa}). Furthermore, the reduced-size steady state $v^{\pi}_\infty = \left(\tilde{v}^{\pi}_{\infty, i_1}, ..., \tilde{v}^{\pi}_{\infty, i_r}\right)^T$ is an equilibrium of (\ref{kjnksdfsdfsdf}). Thus, if $\tilde{v}(t)=v_\infty$, then the dynamics of the projection $v_\text{ker}(t)$ obey the linear time-\textit{invariant} (LTI) system
 \begin{align}\label{ljnkrgdergrgerreg}
\frac{d v_\textup{ker}(t)}{dt} = - \left( S+ \operatorname{diag}\left( B v_\infty \right) \right)v_\textup{ker}(t).
\end{align}
Thus, the affine subspace $\left\{v_\infty + v_\text{ker}\big| v_\text{ker} \in \operatorname{ker}(B) \right\}$ is an invariant set of NIMFA, on which the viral dynamics are linear.

Loosely speaking, Theorem~\ref{theorem:decompose_dynamics} shows that a crucial challenge for solving NIMFA on graphs with equitable partitions is the dynamics of the projection $\tilde{v}(t)$, since solving the set of nonlinear equations (\ref{kjnksdfsdfsdf}) seems more difficult than solving the linear time-varying system (\ref{sadfsdgadfdddaaa}) for a given $\tilde{v}(t)$. For a complete graph, the solution $\tilde{v}(t)$ to set of nonlinear equations (\ref{kjnksdfsdfsdf}) is one-dimensional and can be stated in closed form \cite{van2014sis}. Thus, we obtain the solution of NIMFA on the complete graph, for \textit{arbitrary} initial viral states $v(0)$, as:

\begin{theorem}\label{theorem:complete_graph_solution}
Consider NIMFA (\ref{NIMFA_stacked}) on the complete graph, whose infection rates equal $\beta_{ij}=\beta$ for all nodes $i,j=1, ..., N$. Suppose the curing rates satisfy $\delta_i=\delta$ for all nodes $i$. Then, for any initial viral state $v(0)\in [0,1]^N$, the solution of NIMFA (\ref{NIMFA_stacked}) equals
\begin{align}
v(t) = c_1(t) v_\infty  + c_2(t) v_\textup{ker}(0),
\end{align}
where the steady-state vector equals $v_\infty=\left( 1 - \frac{\delta}{\beta N} \right) u$, and the $N\times 1$ vector $v_\textup{ker}(0)$ is given by
\begin{align}
v_\textup{ker}(0) &= \left( I - \frac{1}{N}u u^T\right) v(0).
\end{align}
The scalar function $c_1(t)$ equals
\begin{align}\label{c_1_closedform}
c_1(t) =\frac{1}{2}\left( 1 + \tanh\left( \frac{w}{2} t +\Upsilon_1(0)  \right) \right)
\end{align}
with the \emph{viral slope} $w=\beta N - \delta$ and the constant
\begin{align}
\Upsilon_1(0) = \operatorname{arctanh}\left( 2 \frac{v^T_\infty v(0)}{\lVert v_\infty \rVert^2_2} - 1\right),
\end{align}
and the scalar function $c_2(t)$ equals
\begin{align}\label{c_2_closedform}
c_2(t) = \Upsilon_2(0) e^{- \Phi t } \operatorname{sech}\left( \frac{w}{2} t + \Upsilon_1(0)\right)
\end{align}
with the constant $\Phi = \beta N v_{\infty, i}/2 +\delta$, for an arbitrary node $i$, and the constant
\begin{align}\label{kjbnkjnsdfaserwfaerfaedrfaerf}
\Upsilon_2(0) = \frac{v^T_\textup{ker}(0) v(0)}{\lVert v_\textup{ker}(0) \rVert^2_2} \cosh\left( \Upsilon_1(0)\right) .
\end{align}
\end{theorem}
\begin{proof}
Appendix~\ref{appendix:complete_graph_solution}.
\end{proof}

\begin{figure}[!ht]
\captionsetup[subfigure]{justification=centering}
    \centering
     \begin{subfigure}[t]{0.99\textwidth}
         \centering
              \includegraphics[width=\textwidth]{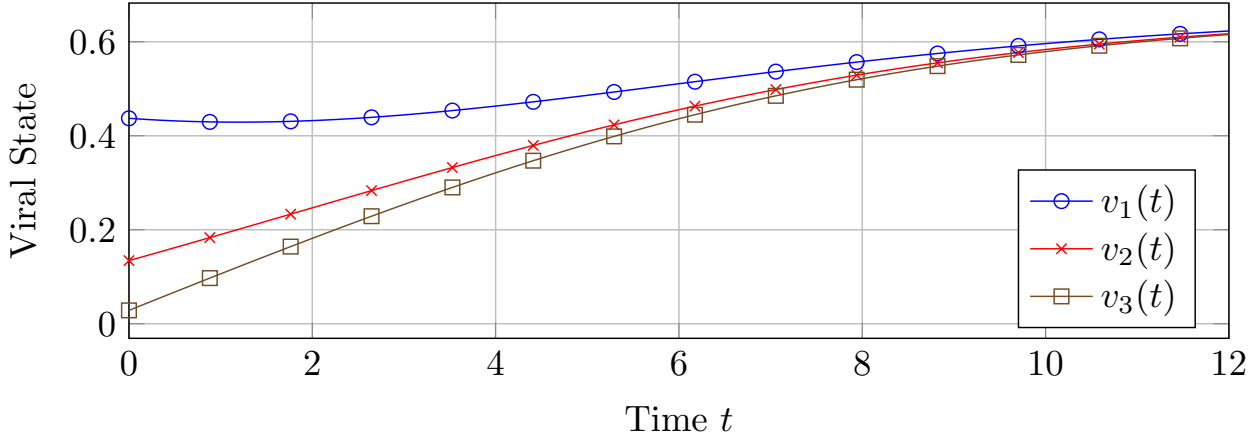}
         \caption{Viral state $v(t)$ versus time $t$.}
     \end{subfigure} 
     \vspace{3mm}     \\
     \begin{subfigure}[t]{0.99\textwidth}
         \centering
              \includegraphics[width=\textwidth]{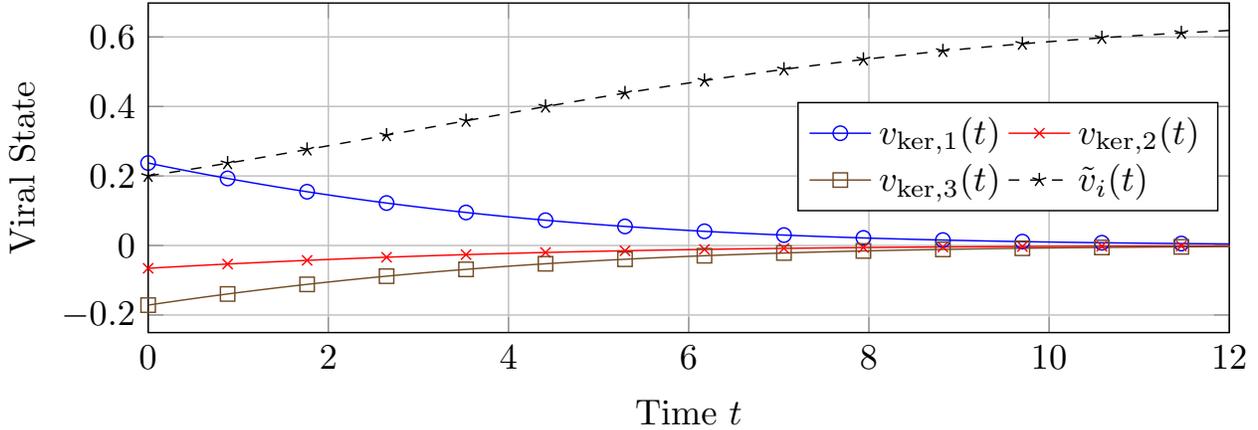}
         \caption{Projections $\tilde{v}(t)$ and $v_\text{ker}(t)$ versus time $t$.}
     \end{subfigure}     
\caption{\textbf{Closed-form solution of NIMFA on the complete graph.} The solution of NIMFA (\ref{NIMFA_continuous}) for a complete graph with $N=3$ nodes and homogeneous spreading rates. As stated by Theorem~\ref{theorem:decompose_dynamics}, the viral state satisfies $v(t)=\tilde{v}(t) + v_{\text{ker}}(t)$, where $\tilde{v}(t)$ and $v_{\text{ker}}(t)$ denote the projection of the viral state $v(t)$ on the subspace $\mathcal{V}_{\neq 0}$ and the kernel $\operatorname{ker}(B)$, respectively. \textbf{(a)}: The viral state $v_i(t)$ versus time $t$ for every node $i$. \textbf{(b)}: The projections $\tilde{v}(t)$ and $v_{\text{ker}}(t)$, which follow from Theorem~\ref{theorem:complete_graph_solution} as $\tilde{v}_i(t) = c_1(t) v_{\infty, i}$ and $v_{\text{ker},i}(t) = c_2(t) \left( y_2\right)_i$ for all nodes $i$, where the scalar functions $c_1(t)$ and $c_2(t)$ are given by the closed-form expressions (\ref{c_1_closedform}) and (\ref{c_2_closedform}), respectively. Since the steady state $v_{\infty, i}$ is the same for every node $i$ in the complete graph, it holds that $\tilde{v}_i(t) = \tilde{v}_j(t)$ for all nodes $i,j$.}\label{fig:kkkqqooii}
\end{figure}

Figure~\ref{fig:kkkqqooii} illustrates the closed-form solution of NIMFA for complete graphs, as given by Theorem~\ref{theorem:complete_graph_solution}. As shown by Figure~\ref{fig:kkkqqooii}, even though the viral state average $\tilde{v}(t)$ is monotonically increasing, the viral state $v_1(t) = \tilde{v}_1(t) + v_{\text{ker}, 1}(t)$ is decreasing until $t\approx 1$, which is due to the dynamics of the projection $v_\text{ker}(t)$ on the kernel $\operatorname{ker}(B)$.

\section{Approximate clustering}

As shown by Theorem~\ref{theorem:invariant_sets_are_equitable_partitions}, equitable partitions and low-dimensional viral state dynamics in NIMFA are equivalent. Many networks possess some macroscopic structure, which may \textit{resemble} an equitable partition, but which \textit{is not precisely} an equitable partition. \textit{Is it possible to reduce the number of NIMFA equations, if the network has an \quotes{almost} equitable partition?}

For two $N\times 1$ vectors~$x,y$, $x\ge y$ denotes that $x_i\ge y_i$ for all entries~$i=1, ..., N$. Theorem~\ref{theorem:NIMFA_bounds} shows that NIMFA (\ref{NIMFA_stacked}) on any network can be bounded by increasing or decreasing the spreading rates $\beta_{ij}$,$\delta_i$:
\begin{theorem}\label{theorem:NIMFA_bounds}
Consider two NIMFA systems with respective positive curing rates~$\delta_i$ and $\tilde{\delta}_i$, non-negative infection rates~$\beta_{ij}$ and $\tilde{\beta}_{ij}$, and viral states~$v_i(t)$ and $\tilde{v}_i(t)$. Suppose that the initial viral state $v_i(0)$,$\tilde{v}_i(0)$ are in $[0,1]$ for all nodes $i$ and that the matrices $B$ and $\tilde{B}$, with elements $\beta_{ij}$ and $\tilde{\beta}_{ij}$, respectively, are irreducible. Then, if $\tilde{\delta}_i \le \delta_i$ and $\tilde{\beta}_{ij} \ge \beta_{ij}$ for all nodes $i,j$, $\tilde{v}(0)\ge v(0)$ implies that~$\tilde{v}(t)\ge v(t)$ at every time $t$.
\end{theorem}
\begin{proof}
Appendix~\ref{appendix:NIMFA_bounds}.
\end{proof}
We emphasise that Theorem~\ref{theorem:NIMFA_bounds} does not assume symmetric infection rate matrices $B$, $\tilde{B}$. Building upon Theorem~\ref{theorem:NIMFA_bounds}, we aim to bound the viral state $v(t)$ of any network at every time $t$ by the viral state of networks with equitable partitions. In the following, we consider a partition $\pi=\{\mathcal{N}_1,...,\mathcal{N}_r\}$ of the node set $\mathcal{N}=\{1, ..., N\}$ of an arbitrary network. We stress that $\pi$ can be \textit{any}, not necessarily equitable, partition. We define the minimum $d_{\textup{min},pl}$ of the sum of infection rates from cell $\mathcal{N}_l$ to $\mathcal{N}_p$ as
\begin{align}\label{d_lowwwww}
d_{\textup{min},pl} = \underset{i\in\mathcal{N}_p}{\operatorname{min}} ~ \sum_{k \in \mathcal{N}_l} \beta_{ik} 
\end{align}
and the maximum $d_{\textup{max},pl}$ as
\begin{align}\label{d_uppp}
d_{\textup{max},pl} = \underset{i\in\mathcal{N}_p}{\operatorname{max}} ~ \sum_{k \in \mathcal{N}_l} \beta_{ik}. 
\end{align}
Furthermore, we denote the $r\times r$ matrices $B_\text{min}$ and $B_\text{max}$, whose elements are given by $d_{\textup{min},pl}$ and $d_{\textup{max},pl}$, respectively. Analogously, we define the minimum $\delta_{\text{min},l}$ of the curing rates in cell $\mathcal{N}_l$ as
\begin{align}\label{delta_lb_def}
\delta_{\text{min},l} = \underset{i\in\mathcal{N}_l}{\operatorname{min}} ~ \delta_i
\end{align}
and the maximum $\delta_{\text{max},l}$ as
\begin{align}\label{delta_ub_def}
\delta_{\text{max},l} = \underset{i\in\mathcal{N}_l}{\operatorname{max}} ~ \delta_i.
\end{align}
We combine Theorem~\ref{theorem:equitableOriginal} and Theorem~\ref{theorem:NIMFA_bounds} to obtain:
\begin{theorem}\label{theorem:bound_by_equitable}
Suppose that the Assumptions~\ref{assumption:spreading_rates} and~\ref{assumption:matrix_B_symmetric} hold. At every time $t$, consider the $r\times 1$ reduced-size lower bound $v_{\textup{lb},l}(t)$ and $r\times 1$ upper bound $v_{\textup{ub},l}(t)$, which evolve as
\begin{align}\label{wrfaergaergaergrea}
	\frac{d v_{\textup{lb}}(t)}{d t } & = - \operatorname{diag}\left(\delta_{\textup{max},1},..., \delta_{\textup{max},r}\right) v_{\textup{lb}}(t) + \textup{\textrm{diag}}\left( u_r - v_{\textup{lb}}(t)\right) B_{\textup{min}}  v_{\textup{lb}}(t)
	\intertext{and} 
	\frac{d v_{\textup{ub}}(t)}{d t } & = - \operatorname{diag}\left(\delta_{\textup{min},1},..., \delta_{\textup{min},r}\right) v_{\textup{ub}}(t) + \textup{\textrm{diag}}\left( u_r - v_{\textup{ub}}(t)\right) B_{\textup{max}}  v_{\textup{ub}}(t).
\end{align}
Then, if the initial states satisfy $v_{\textup{lb},l}(0) \le v_i(0)\le v_{\textup{ub},l}(0)$ for all nodes $i$ in any cell $\mathcal{N}_l$, the viral state $v_i(t)$ of all nodes $i$ in any cell $\mathcal{N}_l$ is bounded by
\begin{align}\label{khbkervgergeg}
v_{\textup{lb},l}(t)\le v_i(t) \le v_{\textup{ub},l}(t) \quad \forall t\ge 0.
\end{align}
\end{theorem}
\begin{proof}
Appendix \ref{appendix:bound_by_equitable}.
\end{proof}

Theorem~\ref{theorem:bound_by_equitable} states that the $N\times 1$ viral state $v(t)$ on \textit{any} network is bounded by the $r\times 1$ viral states $v_{\textup{lb}}(t)$, $v_{\textup{ub}}(t)$ on networks with equitable partitions and $r$ cells. Reducing the $N$-dimensional viral state dynamics to $r$-dimensional dynamics comes at the cost of an \textit{approximate} description by the bounds in (\ref{khbkervgergeg}). If the partition $\pi$ is equitable, then it holds that $d_{\text{min},pl}=d_{\text{max},pl}$, and the bounds in Theorem~\ref{theorem:bound_by_equitable} can be replaced by the exact statement in Theorem~\ref{theorem:equitableOriginal}.

Similarly to the lower bound and upper bound of the degrees in (\ref{d_lowwwww}) and (\ref{d_uppp}), respectively, we define the \textit{average} degree from cell $\mathcal{N}_l$ to $\mathcal{N}_p$ for any partition $\pi$ as
\begin{align}
\bar{d}_{pl} = \frac{1}{\big|\mathcal{N}_p\big|}
\sum_{i \in \mathcal{N}_p} \sum_{k \in \mathcal{N}_l} \beta_{ik}.
\end{align}
Then, we define the $r\times r$ reduced-size infection rate matrix $\bar{B}$, which consists of the elements $\bar{d}_{pl}$. Furthermore, we define the average curing rate of any cell $\mathcal{N}_l$ as
\begin{align}
\bar{\delta}_l = \frac{1}{\big|\mathcal{N}_l\big|}\sum_{i \in \mathcal{N}_l} \delta_i.
\end{align}
Then, we approximate the viral state by $v_i(t)\approx \bar{v}_l(t)$ for all nodes $i$ in any cell $\mathcal{N}_l$. Here, the $r\times 1$ reduced-size viral state vector $\bar{v}(t)$ evolves as 
\begin{align}\label{khbjegvetgetbtbrtbrt}
	\frac{d \bar{v}(t)}{d t } & = - \operatorname{diag}\left(\bar{\delta}_1,..., \bar{\delta}_r\right) \bar{v}(t) + \textup{\textrm{diag}}\left( u_r - \bar{v}(t)\right) \bar{B} \bar{v}(t),
\end{align}
and, for all cells $\mathcal{N}_l$, the initial state equals
\begin{align}
\bar{v}_l(0) = \frac{1}{\big|\mathcal{N}_l\big|}\sum_{i \in \mathcal{N}_l} v_l(0).
\end{align}
If the matrix $B$ has an equitable partition $\pi$ and the rates $\delta_i$, $\beta_{ij}$ are the same between all nodes $i,j$ in any two cells as in Theorem~\ref{theorem:decompose_dynamics}, then the approximation $\bar{v}(t)$ coincides with the projection $\tilde{v}(t)$ of the viral state $v(t)$ on the subspace $\mathcal{V}_{\neq 0}$.

To illustrate the accuracy of the bounds in Theorem~\ref{theorem:bound_by_equitable} and the reduced-size viral state $\bar{v}(t)$ for networks without equitable partitions, we consider the \textit{Stochastic Blockmodel} (SBM), originally introduced by Holland \textit{et al.} \cite{holland1983stochastic}. We consider a network with $N=1000$ nodes and a partition $\pi$ with $r=5$ cells $\mathcal{N}_1$, ..., $\mathcal{N}_5$. The cells are of size $|\mathcal{N}_1|=400$, $|\mathcal{N}_2|=250$, $|\mathcal{N}_3|=200$, $|\mathcal{N}_4|=100$ and $|\mathcal{N}_5|=50$. With a probability of $0.7$, there are no links between two cells $\mathcal{N}_p$, $\mathcal{N}_l$, i.e., $\beta_{ij}=\beta_{ji}=0$ for all nodes $i\in \mathcal{N}_p$ and $j\in\mathcal{N}_l$. Otherwise, with a probability of $0.3$, we denote the mean of the links between the cells $\mathcal{N}_p$, $\mathcal{N}_l$ by $\bar{\beta}_{pl}=\bar{\beta}_{lp}$, which is set to a uniform random number in $[0.1,0.2]$. Then, the infection rate $\beta_{ij}=\beta_{ji}$ for all nodes $i\in \mathcal{N}_p$ and $j\in\mathcal{N}_l$ is set to a random number $[\bar{\beta}_{pl}, \bar{\beta}_{pl}(1+\sigma_\text{rel})]$, where we vary the \textit{relative variance} $\sigma_\text{rel}$ for different scenarios in the numerical evaluation. If $\sigma_\text{rel}=0$, then the partition $\pi$ is equitable. The larger the variance $\sigma_\text{rel}$, the \quotes{less equitable} the partition $\pi$. For every node $i$, the curing rate $\delta_i$ is set to a uniform random number in $[1, 1+\sigma_\text{rel}]$, and the initial viral state $v_i(0)$ is set to a uniform random number in $[0.01, 0.01(1+\sigma_\text{rel})]$. Hence, if the variance $\sigma_\text{rel}=0$, then it holds that 
$v_{\textup{lb},l}(t)=v_{\textup{lb},l}(t)=v_i(t)$ for every node $i$ in any cell $\mathcal{N}_l$. Lastly, the curing rates are decreased to $\delta_i\gets c \delta_i$, where the scalar $c$ is chosen such that the basic reproduction number (\ref{R_0_NIMFA_def}) equals $R_0=3$. To obtain the viral state $v(t)$, we discretise NIMFA (\ref{NIMFA_continuous}) with a sufficiently small sampling time, see \cite{pare2018analysis, prasse2019viral, liu2020stability} for a detailed analysis of the resulting discrete-time NIMFA model.

\begin{figure}[!ht]
\captionsetup[subfigure]{justification=centering}
    \centering
     \begin{subfigure}[t]{0.48\textwidth}
         \centering
         \includegraphics[width=\textwidth]{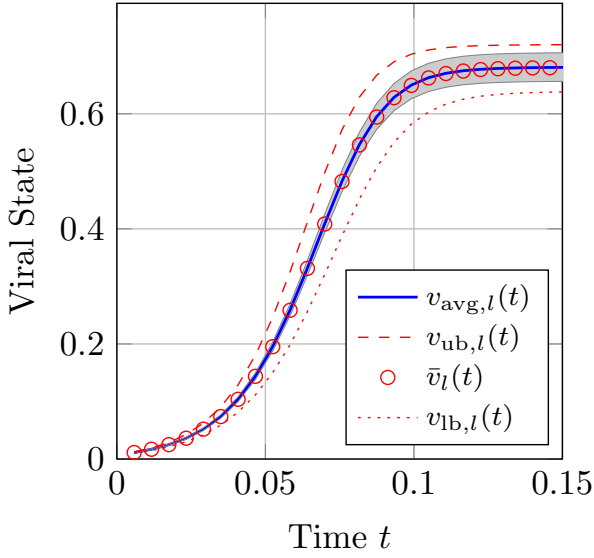}
         \caption{Cell $\mathcal{N}_1$ and relative variance $\sigma_\text{rel}=0.25$.}
     \end{subfigure} 
     \quad
         \begin{subfigure}[t]{0.48\textwidth}
         \centering
         \includegraphics[width=\textwidth]{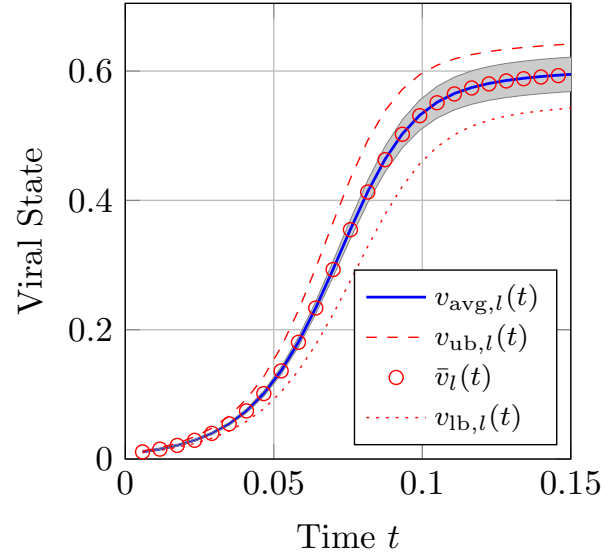}
         \caption{Cell $\mathcal{N}_5$ and relative variance $\sigma_\text{rel}=0.25$.}
     \end{subfigure} 
     \vspace{2.5mm}\\
     \begin{subfigure}[t]{0.48\textwidth}
         \centering
         \includegraphics[width=\textwidth]{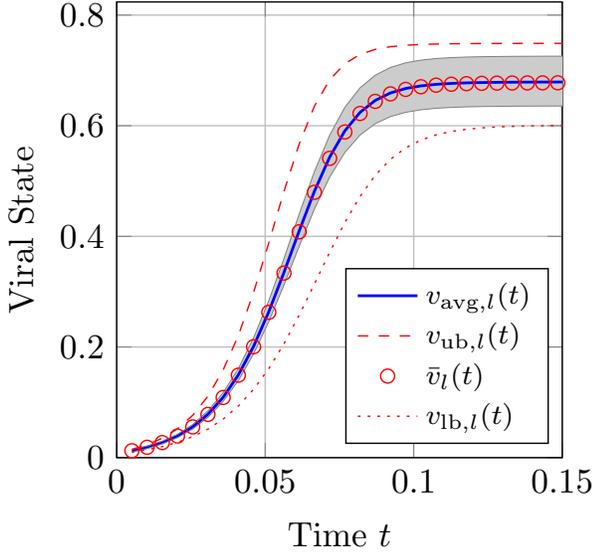}
         \caption{Cell $\mathcal{N}_1$ and relative variance $\sigma_\text{rel}=0.5$.}
     \end{subfigure} 
     \quad
         \begin{subfigure}[t]{0.48\textwidth}
         \centering
         \includegraphics[width=\textwidth]{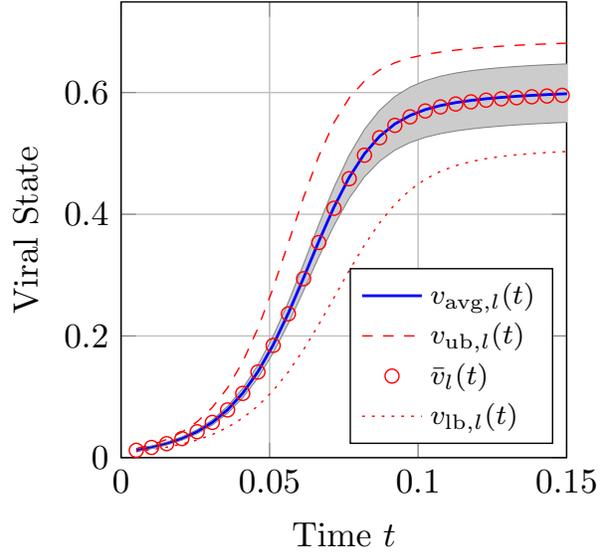}
         \caption{Cell $\mathcal{N}_5$ and relative variance $\sigma_\text{rel}=0.5$.}
     \end{subfigure} 
     \caption{\textbf{Low-dimensional approximation of the viral state dynamics.} For a stochastic blockmodel network with $N=1000$ nodes and $r=5$ cells, the accuracy of the approximation $\bar{v}_l(t)$ and the tightness of the bounds $v_{\textup{lb},l}(t)$, $v_{\textup{lb},l}(t)$ are depicted. The reduced-size viral states $\bar{v}(t)$,$v_{\textup{lb}}(t)$ and $v_{\textup{lb}}(t)$ are equal to the linear combination of $m=r=5$ agitation modes $y_l$, each of which corresponds to one cell. The first and second row correspond to the relative variance $\sigma_\text{rel}=0.25$ and $\sigma_\text{rel}=0.5$, respectively. The left column corresponds to the largest cell $\mathcal{N}_1$, the right column corresponds to the smallest cell $\mathcal{N}_5$. The viral state $v_i(t)$ of every node $i$ in the respective cell $\mathcal{N}_l$ is within the shaded grey area.}
  \label{fig:chapter1_bounds_1}
\end{figure}   

Figure~\ref{fig:chapter1_bounds_1} illustrates the accuracy of the bounds $v_{\textup{lb},l}(t)$, $v_{\textup{lb},l}(t)$ in Theorem~\ref{theorem:bound_by_equitable} and the approximation accuracy of $\bar{v}(t)$ in (\ref{khbjegvetgetbtbrtbrt}) for the largest cell $\mathcal{N}_1$ and the smallest cell $\mathcal{N}_5$. For both $\sigma_\text{rel}=0.25$ and $\sigma_\text{rel}=0.5$, the approximation $\bar{v}_l(t)$ is close to the exact average viral state in cell $\mathcal{N}_l$,
\begin{align}\label{ljnergvesrgstbtgb}
v_{\text{avg}, l}(t) = \frac{1}{\big|\mathcal{N}_l\big|}\sum_{i \in \mathcal{N}_l} v_l(t).
\end{align}
The accuracy of the bounds $v_{\textup{lb},l}(t)$, $v_{\textup{lb},l}(t)$ on any viral state $v_i(t)$ in cell $\mathcal{N}_l$ decreases when the variance $\sigma_\text{rel}$ is increased. Nonetheless, the bounds $v_{\textup{lb},l}(t)$, $v_{\textup{lb},l}(t)$ are reasonably accurate for both $\sigma_\text{rel}=0.25$ and $\sigma_\text{rel}=0.5$.

\subsection{Clustering for epidemics on real-world networks}

Approximating the viral state dynamics by $m<N$ equations requires the specification of a partition $\pi$ of the nodes. In some cases, this partition is given \textit{a priori}, as in the experiments in Figure~\ref{fig:chapter1_bounds_1}, where the node partition $\pi$ was chosen corresponding to the SBM blocks. In contrast, for real-world networks, it is more challenging to determine an appropriate clustering and, hence, to obtain an accurate description of the viral state dynamics by $m<N$ equations.

We consider a two-step approach to reduce NIMFA to $m=r<N$ equations. First, we obtain a partition $\pi$ of the nodes by the Bethe spectral clustering algorithm \cite{NIPS2014_63923f49}, which makes use of the \textit{Bethe Hessian} $H_{\pm} = (d_\text{avg}-1)I \pm d_\text{avg} B+D$, with the average degree $d_\text{avg}$ and the degree matrix $D=\operatorname{diag}(d_1, ..., d_N)$. When the matrix $B$ has an (approximate) SBM structure, the negative eigenvalues of $H_{\pm}$ have corresponding eigenvectors which are (approximately) piecewise constant on the blocks of $B$. The spectral clustering algorithm partitions the nodes of $B$ based on a $k$-means clustering of the negative eigenvector entries of $H_{\pm}$. Second, we evaluate the accuracy of reduced-size viral state $\bar{v}(t)$ in (\ref{khbjegvetgetbtbrtbrt}) by the deviation of the prevalence, 
\begin{align}\label{ejnjnkjnkjbkb}
\epsilon_{\text{avg}} = \sum^n_{k=1}\left|  \frac{1}{N}  \sum^N_{i=1} v_i\left(k\Delta t\right)  -  \frac{1}{N}  \sum^r_{l=1} \left|\mathcal{N}_l\right| \bar{v}_l\left(k\Delta t\right) \right|.
\end{align}
Here, $\Delta t$ denotes the sampling time, $k$ is the discrete time, and the number of observations $n$ is chosen such that the viral state $v(n\Delta t)$ practically converged to the steady state $v_\infty$.

We applied the Bethe clustering algorithm to three real-world networks, which were accessed through \cite{kunegis2013konect}: the \textit{American football} network \cite{girvan2002community} with $N=115$ nodes and $L=613$ links, for which $r=10$ clusters were detected; the \textit{primary school} contact network (day 1) \cite{10.1371/journal.pone.0023176} with $N=236$ nodes and $L=5899$ links, resulting in $r=8$ clusters; and the \textit{train bombing} network \cite{hayes2006connecting} with $N=64$ nodes, $L=243$ links and $r=3$ identified clusters. For all networks, we considered homogeneous spreading rates $\beta_{ij}$, $\delta_i$, which were set such that the basic reproduction number equals $R_0=3$. The initial viral state was set to $v_i(\Delta t)=1/N$ for every node $i$. To evaluate the accuracy of the Bethe clustering approach, we additionally considered a collection of random partitions, which are obtained by randomly permuting the nodes in the partition $\pi$ of the Bethe clustering.

\begin{figure}[!ht]
\captionsetup[subfigure]{justification=centering}
    \centering
         \begin{subfigure}[t]{0.3\textwidth}
         \centering
         \includegraphics[width=\textwidth]{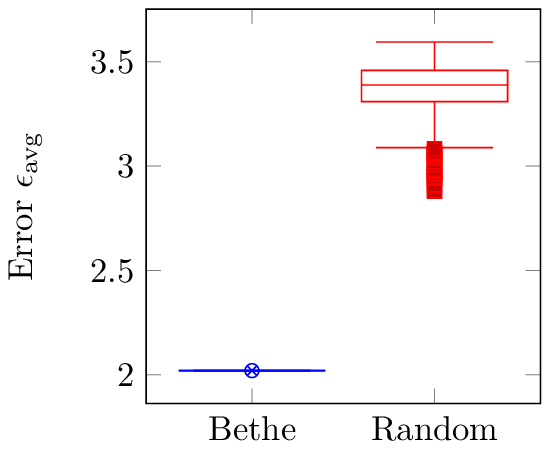}
         \caption{American football.}
     \end{subfigure} 
     \quad
     \begin{subfigure}[t]{0.3\textwidth}
         \centering
         \includegraphics[width=\textwidth]{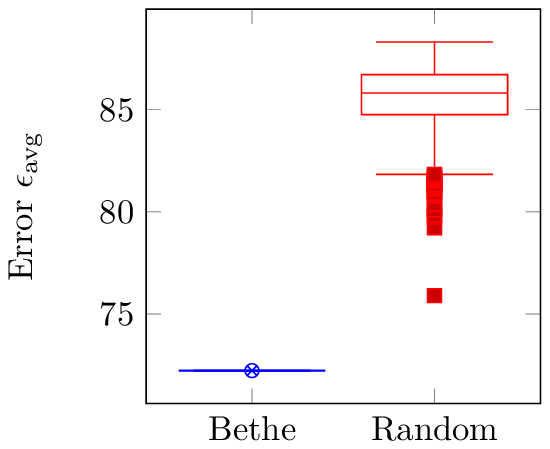}
         \caption{Primary school.}
     \end{subfigure} 
     \quad
         \begin{subfigure}[t]{0.3\textwidth}
         \centering
         \includegraphics[width=\textwidth]{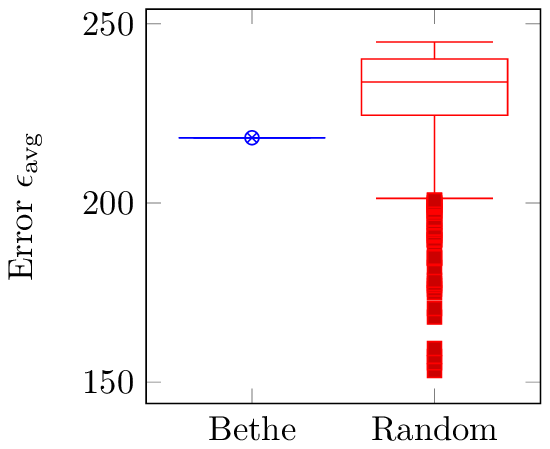}
         \caption{Train bombing.}
     \end{subfigure} 
     \caption{\textbf{Low-dimensional approximation of epidemics on real-world networks.} The error $\epsilon_{\text{avg}}$ of the reduced-size viral state $\bar{v}(t)$, in (\ref{khbjegvetgetbtbrtbrt}), for partitions obtained by Bethe clustering and random partitions.}
  \label{fig:mean_sbm_like_networks}
\end{figure}   

Figure~\ref{fig:mean_sbm_like_networks} shows that, for the football and the school network which have a clear community structure, the Bethe spectral clustering approach results in significantly more accurate low-dimensional viral dynamics $\bar{v}(t)$ than for random partitions. For the train network, which does not possess a clear community structure, there is a smaller advantage of Bethe clustering. Thus, our results indicate that if the network has an underlying community structure, then spectral clustering may be used to find an accurate low-dimensional approximation of the viral state dynamics.

\begin{figure}[!ht]
\captionsetup[subfigure]{justification=centering}
    \centering
         \begin{subfigure}[t]{0.3\textwidth}
         \centering
         \includegraphics[width=\textwidth]{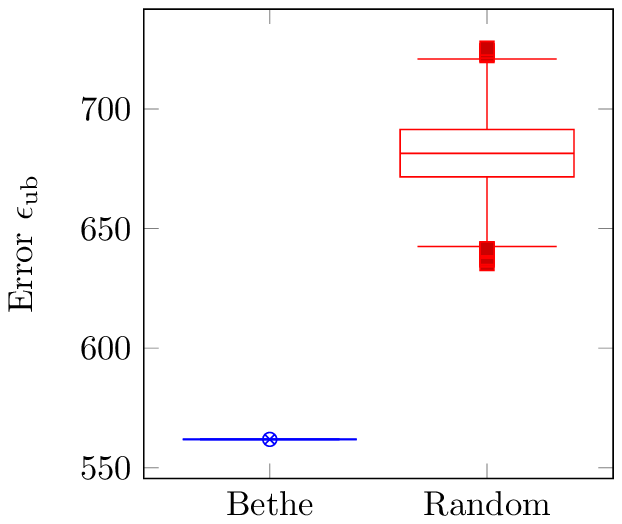}
         \caption{American football.}
     \end{subfigure} 
     \quad
     \begin{subfigure}[t]{0.3\textwidth}
         \centering
         \includegraphics[width=\textwidth]{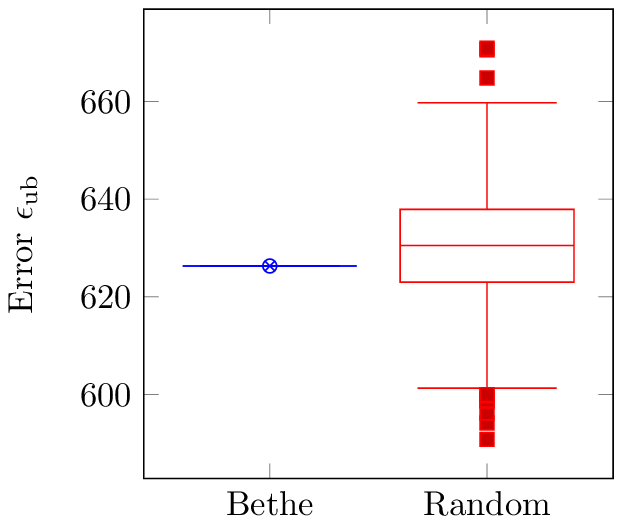}
         \caption{Primary school.}
     \end{subfigure} 
     \quad
         \begin{subfigure}[t]{0.3\textwidth}
         \centering
         \includegraphics[width=\textwidth]{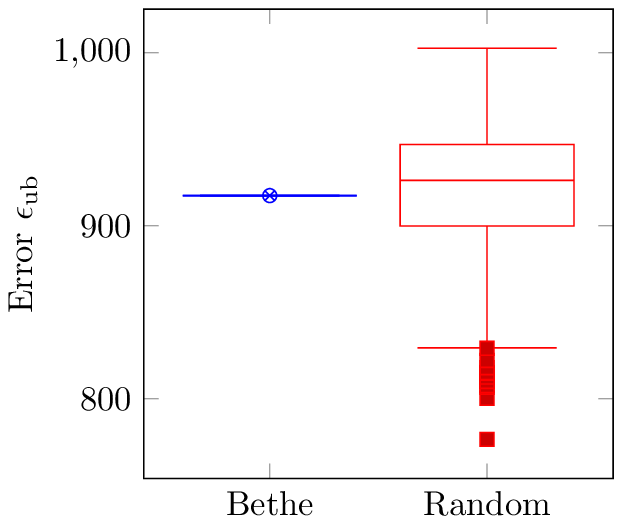}
         \caption{Train bombing.}
     \end{subfigure} 
          \vspace{2.5mm}\\
         \begin{subfigure}[t]{0.3\textwidth}
         \centering
         \includegraphics[width=\textwidth]{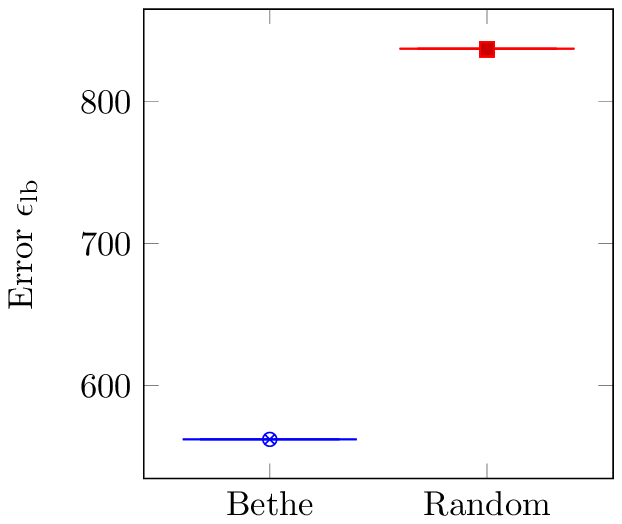}
         \caption{American football.}
     \end{subfigure} 
     \quad
          \begin{subfigure}[t]{0.3\textwidth}
         \centering
         \includegraphics[width=\textwidth]{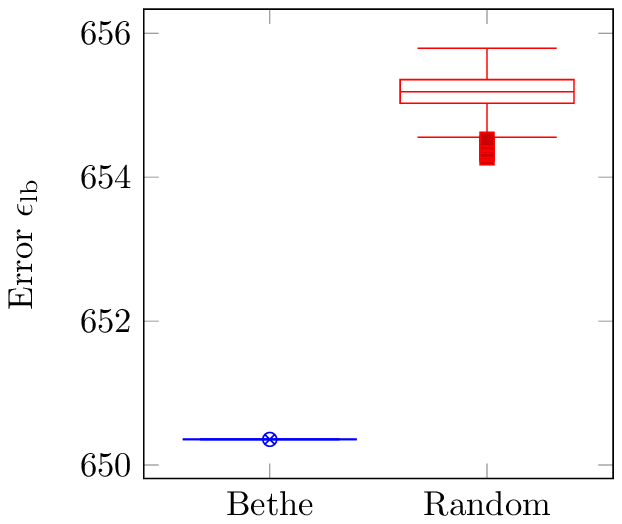}
         \caption{Primary school.}
     \end{subfigure} 
     \quad
         \begin{subfigure}[t]{0.3\textwidth}
         \centering
         \includegraphics[width=\textwidth]{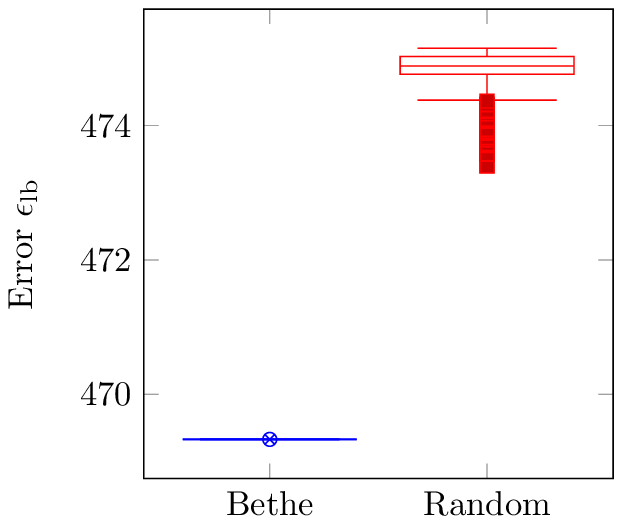}
         \caption{Train bombing.}
     \end{subfigure} 
     \caption{\textbf{Low-dimensional bounds of epidemics on real-world networks.} The errors of the low-dimensional bounds $v_{\textup{lb},l}(t)$ and $v_{\textup{ub},l}(t)$, stated by Theorem~\ref{theorem:bound_by_equitable}, for partitions obtained by Bethe clustering and random partitions. The subplots in the first and second row show the errors $\epsilon_{\text{ub}}$ and $\epsilon_{\text{lb}}$ of the upper bound $v_{\textup{ub},l}(t)$ and the lower bound $v_{\textup{ub},l}(t)$, respectively.}
  \label{fig:bounds_sbm_like_networks}
\end{figure}  

Furthermore, for any partition $\pi$ of the nodes, there are low-dimensional bounds $v_{\textup{lb},l}(t)$, $v_{\textup{ub},l}(t)$ of the viral state dynamics, as stated by Theorem~\ref{theorem:bound_by_equitable}. We define the errors $\epsilon_{\text{ub}}$ and $\epsilon_{\text{lb}}$ of the bounds $v_{\textup{ub},l}(t)$ and $v_{\textup{lb},l}(t)$ analogously to (\ref{ejnjnkjnkjbkb}). Figure~\ref{fig:bounds_sbm_like_networks} demonstrates that
the partition of the nodes by the Bethe clustering algorithm results in significantly more accurate lower bounds $v_{\textup{lb},l}(t)$ than those obtained from random partitions, and somewhat more accurate upper bounds $v_{\textup{ub},l}(t)$.

\section{Conclusions}
In this work, we focussed on reducing NIMFA on a network with $N$ nodes to only $m<<N$ differential equations. We believe that the geometric clustering approach outlined in this work can be applied to other dynamics on networks, particularly to general epidemic models \cite{sahneh2013generalized,prasse2019gemf} and the class of dynamics in \cite{timme2007revealing, barzel2013universality, laurence2019spectral, prasse2020predicting}. Our contribution is composed of three parts. In the first part, we showed that the viral dynamics evolve on an $m$-dimensional subspace $\mathcal{V}$ \textit{if and only if} the contact network has an equitable partition with $m_1\le m$ cells. Thus, low-dimensional viral state dynamics and the macroscopic structure of equitable partitions are equivalent. 

In the second part, we focussed on equitable partitions $\pi$ with the same spreading rates $\beta_{ij}$ and $\delta_i$ for all nodes $i,j$ in the same cell $\mathcal{N}_l$. We considered the decomposition of the viral state $v(t)=v_{\text{ker}}(t)+\tilde{v}(t)$ into two parts: the term $\tilde{v}(t)$ describes the viral state average in every cell $\mathcal{N}_l$; and the term $v_{\text{ker}}(t)$ equals the projection of the viral state $v(t)$ onto the kernel of the infection rate matrix $B$. By showing that the term $\tilde{v}(t)$ evolves independently from the projection $v_{\text{ker}}(t)$ and the projection $v_{\text{ker}}(t)$ obeys a linear time-varying system, we derived the solution of the NIMFA differential equations on the complete graph for \textit{arbitrary} initial conditions $v(0)$.

Strictly speaking, most contact networks do not have an equitable partition, and an exact reduction of the number of NIMFA equations is not possible. In the third part, we considered arbitrary contact networks with a (not necessarily equitable) partition of the nodes into $m$ cells. For any partition of the nodes, we derived bounds and approximations of the NIMFA epidemics with only $m$ differential equations. The \quotes{more equitable} the partition, the more accurate the approximation. Thus, finding (almost) equitable partitions is crucial for reducing an epidemic outbreak in a large population to the interaction of only few groups of individuals.

\section*{Acknowledgements}

We are grateful to Massimo Achterberg for helpful discussions on this material.

\newpage


\appendix

\section{Proof of Lemma~\ref{lemma:invariance}}\label{appendix:lemma_invariance}
Let $w$ denote a vector in the orthogonal complement $\mathcal{V}^\bot$ of the invariant set $\mathcal{V}$. Hence, it must hold that $w^T v(t)=0$ for every time $t \ge 0$ if $v(0)\in \mathcal{V}$, which is equivalent to both $w^T v(0) = 0$ and
 \begin{align}\label{derivativeZero}
 \frac{d ( w^T v(t) ) }{d t} = 0 \quad \forall v(t)\in \mathcal{V}, w \in \mathcal{V}^\bot.
 \end{align} 
We replace the notation $v(t)\in \mathcal{V}$ by $v\in \mathcal{V}$. Then, we obtain from the NIMFA equations (\ref{NIMFA_stacked}) that (\ref{derivativeZero}) is equivalent to
  \begin{align}
    w^T\left(-S v + \operatorname{diag}(u - v) B v \right) = 0 \quad \forall v\in \mathcal{V}, w \in \mathcal{V}^\bot.
  \end{align}   
 Under Assumption~\ref{assumption:delta_v_in_V}, it holds that $S v\in \mathcal{V}$. Hence, the vector $w\in \mathcal{V}^\bot$ is orthogonal to the vector $Sv$, which yields that 
   \begin{align}
    w^T\operatorname{diag}(u - v) B v = 0.
  \end{align}   
  Since $\operatorname{diag}(u)$ is the identity matrix, we obtain that
   \begin{align} \label{ljnljnlss}
    w^T B v =  w^T\operatorname{diag}( v) B v.
  \end{align}   
  Since the invariant set $\mathcal{V}$ is a subspace of~$\mathbb{R}^N$, $v\in \mathcal{V}$ implies that $\gamma v \in \mathcal{V}$ for any scalar $\gamma \in \mathbb{R}$. For the vector~$\gamma v$, where we consider $\gamma>0$, it follows from (\ref{ljnljnlss}) that 
     \begin{align} 
    \gamma w^T B v = \gamma^2 w^T\operatorname{diag}( v) B v,
  \end{align}   
  which is equivalent to
   \begin{align} 
    w^T B v = \gamma w^T\operatorname{diag}( v) B v.
  \end{align}   
Thus, we obtain with (\ref{ljnljnlss}) for every scalar $\gamma>0$ that
   \begin{align}
    w^T\operatorname{diag}( v) B v = \gamma w^T\operatorname{diag}( v) B v,
  \end{align}   
  which implies that 
  \begin{align} \label{quadratic_condition}
    w^T\operatorname{diag}( v) B v =0.
  \end{align} 
  Then, from (\ref{ljnljnlss}), it follows that
   \begin{align}\label{lkjnlsdassa} 
    w^T B v =0
  \end{align} 
  for all vectors $w\in \mathcal{V}^\bot$, $v\in \mathcal{V}$. The vector~$Bv$ is orthogonal to all vectors~$w\in \mathcal{V}^\bot$, only if~$Bv\in\mathcal{V}$. Thus, the set~$\mathcal{V}$ is an invariant subspace~\cite{friedberg1989linear} of the infection rate matrix~$B$. The sets of vectors $y_1, ..., y_m$ and $y_{m+1}, ..., y_N$ span the invariant set~$\mathcal{V}$ and the orthogonal complement~$\mathcal{V}^\bot$, respectively, see (\ref{V_invariant_span}) and (\ref{V_bot_def}). Thus, we can express the symmetric matrix~$B$ as
\begin{align}\label{qwewf}
B = \begin{pmatrix}
y_1 & ... & y_N 
\end{pmatrix} 
\begin{pmatrix}
M_1 & M_{12}\\
0 & M_2
\end{pmatrix}
\begin{pmatrix}
 y^T_1 \\
 \vdots \\
 y^T_N 
\end{pmatrix} 
\end{align}
for some $m\times m$ symmetric matrix $M_1$ and some $(N-m)\times (N-m)$ symmetric matrix $M_2$. The $m\times (N-m)$ matrix~$M_{12}$ describes the mapping from the subspace~$\mathcal{V}^\bot$ to the subspace $\mathcal{V}$. Since the matrix~$B$ is symmetric, it holds that $M_{12}=0$, and (\ref{qwewf}) becomes 
\begin{align} \label{kjnliasccas}
B =\begin{pmatrix}
y_1 & ... & y_N 
\end{pmatrix} 
\begin{pmatrix}
M_1 & 0\\
0 & M_2
\end{pmatrix}
\begin{pmatrix}
 y^T_1 \\
 \vdots\\
 y^T_N 
\end{pmatrix}.
\end{align}
Furthermore, since the matrix~$B$ is diagonalisable as (\ref{diag_B}), the matrices~$M_1$ and $M_2$ are diagonalisable~\cite[Exercise 24, Section 5.4]{friedberg1989linear}. Thus, there is some orthogonal $m\times m$ matrix~$C_1$ and some orthogonal $(N-m)\times (N-m)$ matrix~$C_2$ such that 
\begin{align}\label{ikjnkjnfsdgvdfsag}
B = \begin{pmatrix}
y_1 & ... & y_N 
\end{pmatrix} 
\begin{pmatrix}
C_1 & 0\\
0 & C_2
\end{pmatrix}
\begin{pmatrix}
\Lambda_1 & 0\\
0 & \Lambda_2
\end{pmatrix}
\begin{pmatrix}
C^T_1 & 0\\
0 & C^T_2
\end{pmatrix}
\begin{pmatrix}
 y^T_1 \\
 \vdots \\
 y^T_N 
\end{pmatrix}.
\end{align}
where the $m\times m$ diagonal matrix $\Lambda_1$ and the $(N-m)\times(N-m)$ diagonal matrix $\Lambda_2$ contain the eigenvalues of~$B$. In contrast to the $N\times N$ matrix~$\Lambda$ in (\ref{diag_B}), the diagonal entries of the matrices~$\Lambda_1$ and~$\Lambda_2$ may not be ordered with respect to their magnitude. Hence, there is some permutation~$\phi:\{1, ..., N\} \rightarrow \{1, ..., N\}$ of the eigenvalues~$\lambda_1, ..., \lambda_N$ such that
\begin{align}\label{afdvdfaveargvtbesrtbgtrb}
\Lambda_1 = \operatorname{diag}\left(\lambda_{\phi(1)}, ..., \lambda_{\phi(m)} \right)
\end{align}
and
\begin{align}
\Lambda_2 = \operatorname{diag}\left(\lambda_{\phi(m+1)}, ..., \lambda_{\phi(N)} \right).
\end{align}
We define the $N\times m$ matrix~$E_{\mathcal{V}}$ and the $N\times (N-m)$ matrix~$E_{\mathcal{V}^\bot}$ as
\begin{align} \label{E_V_C}
E_{\mathcal{V}} =\begin{pmatrix}
y_1 & ... & y_m 
\end{pmatrix} C_1
\end{align}
and
\begin{align}
E_{\mathcal{V}^\bot}  =\begin{pmatrix}
y_{N-m} & ... & y_N 
\end{pmatrix} C_2.
\end{align}
Since the matrices $C_1$ and $C_2$ are nonsingular, the columns of the matrices $E_{\mathcal{V}}$ and $E_{\mathcal{V}^\bot}$ span the subspaces $\mathcal{V}$ and $\mathcal{V}^\bot$, respectively. We obtain that
\begin{align}
B = \begin{pmatrix}
E_{\mathcal{V}} & E_{\mathcal{V}^\bot} 
\end{pmatrix} 
\operatorname{diag}\left(\lambda_{\phi(1)}, ..., \lambda_{\phi(N)} \right)
\begin{pmatrix}
E^T_{\mathcal{V}} \\
E^T_{\mathcal{V}^\bot} 
\end{pmatrix}.
\end{align}
Thus, the matrices $E_{\mathcal{V}}, E_{\mathcal{V}^\bot}$ are equal to
\begin{align}\label{E_V}
E_{\mathcal{V}} =\begin{pmatrix}
x_{\phi(1)} & ... & x_{\phi(m)}
\end{pmatrix}
\end{align}
and
\begin{align}\label{E_V_bot}
E_{\mathcal{V}^\bot}  =\begin{pmatrix}
x_{\phi(N-m)} & ... & x_{\phi(N)}
\end{pmatrix},
\end{align}
where the columns $x_{\phi(1)}, ..., x_{\phi(N)}$ are eigenvectors to the eigenvalues $\lambda_{\phi(1)}, ..., \lambda_{\phi(N)}$ of the matrix~$B$, which completes the proof.
\section{Proof of Lemma~\ref{lemma:B_decomposition}}\label{appendix:B_decomposition}
From (\ref{ikjnkjnfsdgvdfsag}), it follows that
\begin{align}
B=\begin{pmatrix}
y_1 & ... & y_m 
\end{pmatrix} 
C_1 \Lambda_1 C^T_1
\begin{pmatrix}
 y^T_1 \\
 \vdots \\
 y^T_m 
\end{pmatrix}
+\begin{pmatrix}
y_{m+1} & ... & y_N
\end{pmatrix} 
C_2 \Lambda_2 C^T_2
\begin{pmatrix}
 y^T_{m+1} \\
 \vdots \\
 y^T_N 
\end{pmatrix}.
\end{align}
We complete the proof by identifying the $m\times m$ matrix $\tilde{B}_{\mathcal{V}}= C_1 \Lambda_1 C^T_1$ and the $(N-m)\times (N-m)$ matrix $\tilde{B}_{\mathcal{V}^\bot}= C_2 \Lambda_2 C^T_2$.

 \section{Proof of Theorem \ref{theorem:invariant_sets_are_equitable_partitions}}
\label{appendix:invariant_sets_are_equitable_partitions}

The proof of Theorem~\ref{theorem:invariant_sets_are_equitable_partitions} is based on four lemmas. First, Lemma~\ref{lemma:w_v_orth} relates the product $\operatorname{diag}(w) v$ to the subspaces~$\mathcal{V}_{\neq 0}$ and $\mathcal{V}^\bot$:
\begin{lemma} \label{lemma:w_v_orth}
For all vectors~$v\in \mathcal{V}_{\neq 0}$ and~$w\in\mathcal{V}^\bot$, it holds that $\operatorname{diag}(w) v \in \mathcal{V}^\bot$.
\end{lemma}
\begin{proof}
Since $w^T\operatorname{diag}(v)=(w_1 v_1, ..., w_N v_N)=v^T\operatorname{diag}(w)$, we obtain from (\ref{quadratic_condition}) that 
 \begin{align}
    v^T\operatorname{diag}(w) B v =0.
  \end{align} 
Equivalently, by taking the transpose, it holds that  
 \begin{align}\label{khbkqaas}
    v^T B \operatorname{diag}(w) v =0.
  \end{align} 
  The invariant set~$\mathcal{V}$ is given by the span of some orthogonal vectors~$y_1, ..., y_m$. By Lemma~\ref{lemma:invariance}, it holds that $\mathcal{V}=\operatorname{span}\{x_{\phi(1)}, ..., x_{\phi(m)}\}$, where $x_{\phi(l)}$ is an eigenvector of the matrix~$B$ to the eigenvalue $\lambda_{\phi(l)}$ for some permutation~$\phi$. Thus, every vector~$v\in \mathcal{V}$ can be written as 
\begin{align}\label{lkjnkljnpasca}
v = \begin{pmatrix}
    x_{\phi(1)} & ... & x_{\phi(m)}
\end{pmatrix} z
\end{align}
for some $m\times 1$ vector $z=(z_1, ..., z_m)^T$, and the subspace~$\mathcal{V}$ equals 
\begin{align}
\mathcal{V}=\left\{  \begin{pmatrix}
    x_{\phi(1)} & ... & x_{\phi(m)}
\end{pmatrix} z \big| z \in \mathbb{R}^m \right\}.
\end{align}
With (\ref{lkjnkljnpasca}), we can rewrite (\ref{khbkqaas}) as
\begin{align}\label{jnwwwuuu}
    z^T \Lambda_1   \begin{pmatrix}
    x^T_{\phi(1)} \\
    \vdots \\
    x^T_{\phi(m)}
\end{pmatrix}     \operatorname{diag}(w)  \begin{pmatrix}
    x_{\phi(1)} & ... & x_{\phi(m)}
\end{pmatrix} z =0,
\end{align}
with the $m\times m$ diagonal matrix~$\Lambda_1=\operatorname{diag}(\lambda_{\phi(1)}, ..., \lambda_{\phi(m)})$. The quadratic form (\ref{jnwwwuuu}) equals zero for all vectors~$cz\in \mathbb{R}^m$ if and only if 
\begin{align}
    \Lambda_1  \begin{pmatrix}
    x^T_{\phi(1)} \\
    \vdots \\
    x^T_{\phi(m)}
\end{pmatrix}      \operatorname{diag}(w) \begin{pmatrix}
    x_{\phi(1)} & ... & x_{\phi(m)}
\end{pmatrix}  =0, 
\end{align}
which implies, with (\ref{lkjnkljnpasca}), that
\begin{align}
    \Lambda_1   \begin{pmatrix}
    x^T_{\phi(1)} \\
    \vdots \\
    x^T_{\phi(m)}
\end{pmatrix}  \operatorname{diag}(w) v =0
\end{align}
for all vectors $v\in\mathcal{V}$. Componentwise, we obtain that 
 \begin{align} \label{ljnkjrgvsdkrfhngbdesrtb}
   \lambda_{\phi(l)}  x^T_{\phi(l)} \operatorname{diag}(w) v = 0 
 \end{align}
for all rows $l=1, ..., m$ and all vectors~$v~\in\mathcal{V}$. Equation (\ref{ljnkjrgvsdkrfhngbdesrtb}) is satisfied if and only if $\lambda_{\phi(l)}=0$ or $x^T_{\phi(l)} \operatorname{diag}(w) v=0$ for all rows~$l=1, ..., m$. The subspace~$\mathcal{V}_0$ contains the vectors~$x_{\phi(l)}$ for which~$\lambda_{\phi(l)}=0$, and the subspace~$\mathcal{V}^\bot$ contains the vectors $x_{\phi(m+1)},...,x_{\phi(N)}$ which are orthogonal to the vectors~$x_{\phi(1)}, ..., x_{\phi(m)}$. Thus, the vector~$\operatorname{diag}(w) v$ must be element of the subspaces~$\mathcal{V}_0$ or~$\mathcal{V}^\bot$, or the vector~$\operatorname{diag}(w) v$ must be equal to the sum of two vectors in the subspaces~$\mathcal{V}_0$ and~$\mathcal{V}$. Hence, with the direct sum (\ref{def:direct_sum}), we can reformulate (\ref{ljnkjrgvsdkrfhngbdesrtb}) as
 \begin{align} \label{kjnrgfksdfadsfadsf}
  \operatorname{diag}(w) v \in \mathcal{V}^\bot \oplus \mathcal{V}_0
 \end{align}
 for all vectors~$v\in\mathcal{V}$. We define the $N\times m_1$ matrix $E_{\mathcal{V}_{\neq 0}}$ as
\begin{align} \label{afasdfasdfasfasdf}
E_{\mathcal{V}_{\neq 0}} =\begin{pmatrix}
x_{\phi(1)} & ... & x_{\phi(m_1)}
\end{pmatrix}
\end{align}
and the $N\times (m - m_1)$ matrix $E_{\mathcal{V}_0}$ as
\begin{align}\label{sdfadfssssss}
E_{\mathcal{V}_0} =\begin{pmatrix}
x_{\phi(m_1+1)} & ... & x_{\phi(m)}
\end{pmatrix}.
\end{align}
Thus, the definition of the matrix $E_{\mathcal{V}}$ in (\ref{E_V}) implies that $E_{\mathcal{V}}= \begin{pmatrix}
 E_{\mathcal{V}_{\neq 0}} & E_{\mathcal{V}_0} 
 \end{pmatrix}$, and the matrix~$\operatorname{diag}(w)$ can be written as
  \begin{align} \label{kjbsdcassd}
  \operatorname{diag}(w) =   \begin{pmatrix}
 E_{\mathcal{V}_{\neq 0}} & E_{\mathcal{V}_0} &  E_{\mathcal{V}^\bot} 
  \end{pmatrix}  
  \begin{pmatrix}
  M_{11} & M_{12} & M_{13}\\
  M_{21} & M_{22} & M_{23}\\
  M_{31} & M_{32} & M_{33}
  \end{pmatrix}
  \begin{pmatrix}
  E^T_{\mathcal{V}_{\neq 0}} \\
  E^T_{\mathcal{V}_0} \\
  E^T_{\mathcal{V}^\bot} 
  \end{pmatrix}   
 \end{align}
for some matrices~$M_{ij}$, where $i,j=1, 2, 3$, whose dimensions follow from the dimension of the matrices~$E_{\mathcal{V}_{\neq 0}}$, $E_{\mathcal{V}_0}$ and $E_{\mathcal{V}^\bot}$. The matrices~$M_{11}$ and~$M_{12}$ describe the mapping of the matrix~$\operatorname{diag}(w)$ from the subspaces~$\mathcal{V}_{\neq 0}$ and~$\mathcal{V}_0$, respectively, to the subspace~$\mathcal{V}_{\neq 0}$. From (\ref{kjnrgfksdfadsfadsf}), we obtain that $M_{11}=0$ and $M_{12} = 0$. Furthermore, since the matrix~$\operatorname{diag}(w)$ is symmetric, it holds that $M_{21}=M^T_{12}=0$. Hence, to satisfy (\ref{kjnrgfksdfadsfadsf}), the matrix~$\operatorname{diag}(w)$ must be equal to
 \begin{align}
  \operatorname{diag}(w) =    \begin{pmatrix}
 E_{\mathcal{V}_{\neq 0}} & E_{\mathcal{V}_0} &  E_{\mathcal{V}^\bot} 
  \end{pmatrix}  
  \begin{pmatrix}
  0& 0 & M_{13}\\
  0 & M_{22} & M_{23}\\
  M_{31} & M_{32} & M_{33}
  \end{pmatrix}
  \begin{pmatrix}
  E^T_{\mathcal{V}_{\neq 0}} \\
  E^T_{\mathcal{V}_0} \\
  E^T_{\mathcal{V}^\bot} 
  \end{pmatrix},   
 \end{align}
 which implies for all vectors~$v\in \mathcal{V}_{\neq 0}$ that $\operatorname{diag}(w) v \in \mathcal{V}^\bot$.
\end{proof}  
 Lemma~\ref{lemma:w_v_orth} states that for all vectors~$v\in \mathcal{V}_{\neq 0}$ and~$w\in\mathcal{V}^\bot$, there must be some vector~$\tilde{w}\in\mathcal{V}^\bot$ such that
 \begin{align} \label{ljknlkregsefadfadf}
  \operatorname{diag}(w) v = \tilde{w}.
 \end{align}
 We aim to find \textit{all} subspaces~$\mathcal{V}_{\neq 0}$ and~$\mathcal{V}^\bot$ whose elements~$v$ and~$w, \tilde{w}$, respectively, satisfy (\ref{ljknlkregsefadfadf}). From Lemma~\ref{lemma:invariance} it follows that a basis of the $N-m$ dimensional subspace~$\mathcal{V}^\bot$ is given by the columns of the matrix
\begin{align} \label{kjnewfwsadvadfsvsdfv}
E_{\mathcal{V}^\bot} =\begin{pmatrix}
\left( x_{\phi(m+1)} \right)_1 &  ... & \left( x_{\phi(N)} \right)_1 \\
\vdots & \ddots & \vdots \\
\left( x_{\phi(m+1)} \right)_N & ... & \left( x_{\phi(N)} \right)_N 
\end{pmatrix}.
\end{align}
For every matrix, the column rank equals the row rank. Since the columns of the matrix~$E_{\mathcal{V}^\bot}$ are linearly independent, there are~$N-m$ linearly independent rows of the matrix~$E_{\mathcal{V}^\bot}$. Without loss of generality\footnote{Otherwise, consider a permutation of the rows, which is equivalent to a relabelling of the nodes.}, we assume that the \textit{first}~$N-m$ rows of the matrix~$E_{\mathcal{V}^\bot}$ are linearly independent. Hence, the first $N-m$ rows span the Euclidean space~$\mathbb{R}^{N-m}$,
\begin{align}\label{kjnarfargergaergarg}
\operatorname{span}
\left\{
\begin{pmatrix}
\left( x_{\phi(m+1)} \right)_1 \\
\vdots\\
 \left( x_{\phi(N)} \right)_1
\end{pmatrix}, 
\begin{pmatrix}
\left( x_{\phi(m+1)} \right)_2 \\
\vdots\\
 \left( x_{\phi(N)} \right)_2
\end{pmatrix},
..., 
\begin{pmatrix}
\left( x_{\phi(m+1)} \right)_{N-m} \\
\vdots\\
 \left( x_{\phi(N)} \right)_{N-m}
\end{pmatrix}
\right\} = \mathbb{R}^{N-m}.
\end{align}
Thus, for \textit{all} vectors~$w\in \mathcal{V}^\bot$ and $v\in\mathcal{V}_{\neq 0}$, there is a vector~$\tilde{w}\in \mathcal{V}^\bot$ whose first~$N-m$ entries satisfy (\ref{ljknlkregsefadfadf}), i.e., 
\begin{align} \label{ljnaerfgaqergasetdgbaswt}
\tilde{w}_i=w_i v_i, \quad i=1, ..., N-m.
\end{align}
The last~$m$ entries of the vector~$\tilde{w}\in \mathcal{V}^\bot$ are determined by the first~$(N-m)$ entries of the vector~$w$, as shown by Lemma~\ref{lemma:w_last_entries_from_first_entries}. (Lemma~\ref{lemma:w_last_entries_from_first_entries} is not a novel contribution, but we include Lemma~\ref{lemma:w_last_entries_from_first_entries} for completeness.)
\begin{lemma}\label{lemma:w_last_entries_from_first_entries}
Suppose that that the first~$N-m$ rows of the matrix~$E_{\mathcal{V}^\bot}$ are linearly independent. Then, there are some $(N-m)\times 1$ vectors~$\chi_{N-m}, ..., \chi_N$ such that the last~$m$ entries of any vector~$w\in \mathcal{V}^\bot$ follow from the first~$(N-m)$ entries as
\begin{align}
w_i= \chi^T_i \begin{pmatrix}
w_1 \\ 
\vdots \\
w_{N-m}
\end{pmatrix}, \quad i=N-m+1, ..., N.
\end{align}
\end{lemma}
\begin{proof}
With the definition of the matrix $E_{\mathcal{V}^\bot}$ in (\ref{kjnewfwsadvadfsvsdfv}), every vector~$w\in\mathcal{V}^\bot$ can be written as 
\begin{align}\label{khbaeferwfareferf}
w =\begin{pmatrix}
x_{\phi(m+1)} &  ... &  x_{\phi(N)} 
\end{pmatrix}
\begin{pmatrix}
z_{m+1}\\
\vdots\\
z_N
\end{pmatrix}
\end{align}
for some scalars $z_{m+1}, ..., z_N\in \mathbb{R}$. Thus, the first $N-m$ entries of the vector~$w$ follow as
\begin{align}\label{khjberkeagvaserg}
\begin{pmatrix}
w_1 \\
\vdots \\
w_{N-m}
\end{pmatrix}= M
\begin{pmatrix}
z_{m+1}\\
\vdots\\
z_N
\end{pmatrix},
\end{align}
where the $(N-m)\times (N-m)$ matrix $M$ equals the first $N-m$ rows of the matrix~$E_{\mathcal{V}^\bot}$,
\begin{align}
M=\begin{pmatrix}
\left( x_{\phi(m+1)}  \right)_1 &  ... & \left( x_{\phi(N)}  \right)_1 \\
\vdots & \ddots & \vdots \\
\left( x_{\phi(m+1)}  \right)_{N-m} & ... & \left( x_{\phi(N)}  \right)_{N-m} 
\end{pmatrix}.
\end{align}
By assumption, the first $N-m$ rows of the matrix~$E_{\mathcal{V}^\bot}$ are linearly independent. Hence, the matrix~$M$ is nonsingular, and the scalars $z_{m+1}, ..., z_N$ follow from (\ref{khjberkeagvaserg}) as
\begin{align}
\begin{pmatrix}
z_{m+1}\\
\vdots\\
z_N
\end{pmatrix}=M^{-1}\begin{pmatrix}
w_1 \\
\vdots \\
w_{N-m}
\end{pmatrix}.
\end{align}
Thus, we obtain the last $m$ entries of the vector~$w$ with (\ref{khbaeferwfareferf}) as
\begin{align}
\begin{pmatrix}
w_{N-m+1} \\
\vdots \\
w_{N}
\end{pmatrix} &=\begin{pmatrix}
\left( x_{\phi(m+1)}  \right)_{N-m+1} &  ... & \left( x_{\phi(N)}  \right)_{N-m+1} \\
\vdots & \ddots & \vdots \\
\left( x_{\phi(m+1)}  \right)_N & ... & \left( x_{\phi(N)}  \right)_N 
\end{pmatrix} \begin{pmatrix}
z_{m+1}\\
\vdots\\
z_N
\end{pmatrix}\\
&=\begin{pmatrix}
\left( x_{\phi(m+1)}  \right)_{N-m+1} &  ... & \left( x_{\phi(N)}  \right)_{N-m+1} \\
\vdots & \ddots & \vdots \\
\left( x_{\phi(m+1)}  \right)_N & ... & \left( x_{\phi(N)}  \right)_N 
\end{pmatrix} M^{-1}\begin{pmatrix}
w_1 \\
\vdots \\
w_{N-m}
\end{pmatrix}.
\end{align}
To complete the proof, we define the vectors~$\chi_{N-m+1}, ..., \chi_N$ as
\begin{align}
\begin{pmatrix}
\chi^T_{N-m+1}\\
\vdots\\
\chi^T_{N}
\end{pmatrix} =\begin{pmatrix}
\left( x_{\phi(m+1)}  \right)_{N-m+1} &  ... & \left( x_{\phi(N)}  \right)_{N-m+1} \\
\vdots & \ddots & \vdots \\
\left( x_{\phi(m+1)}  \right)_N & ... & \left( x_{\phi(N)}  \right)_N 
\end{pmatrix} 
M^{-1}.
\end{align}  
\end{proof}
We combine Lemma~\ref{lemma:w_last_entries_from_first_entries} and (\ref{ljnaerfgaqergasetdgbaswt}), which yields for the last~$(N-m)$ entries of the vector~$\tilde{w}\in \mathcal{V}^\bot$ that
\begin{align}
\tilde{w}_i&= \sum^{N-m}_{j=1} \chi_{ij} \tilde{w}_j \\
&=\sum^{N-m}_{j=1} \chi_{ij} w_j v_j,
\end{align}
where~$i=N-m+1, ..., N$. Furthermore, (\ref{ljknlkregsefadfadf}) states that $\tilde{w}_i=v_i w_i$. Thus, it must hold that
\begin{align}
w_i v_i&=\sum^{N-m}_{j=1} \chi_{ij} w_j v_j
\end{align}
for the entries~$i=N-m+1, ..., N$. Since the vector~$w$ is element of the subspace~$\mathcal{V}^\bot$, we apply Lemma~\ref{lemma:w_last_entries_from_first_entries} again and obtain that
\begin{align}
\left(\sum^{N-m}_{j=1} \chi_{ij} w_j \right) v_i&=\sum^{N-m}_{j=1} \chi_{ij} w_j v_j.
\end{align}
Thus, for all entries~$i=N-m+1, ..., N$, it must hold that
\begin{align} \label{khbergoasrgstg}
\sum^{N-m}_{j=1} \chi_{ij} w_j ( v_i-  v_j)=0
\end{align}
for all vectors~$w\in \mathcal{V}^\bot$ and~$v\in\mathcal{V}_{\neq 0}$. Since the first $N-m$ rows of the matrix~$E_{\mathcal{V}^\bot}$ are linearly independent, see (\ref{kjnarfargergaergarg}), it follows that (\ref{khbergoasrgstg}) must be satisfied for \textit{all} scalars~$w_1$, ..., $w_{N-m}$ in $\mathbb{R}$. Hence, for all vectors~$v\in\mathcal{V}_{\neq 0}$, it must hold that $\chi_{ij} ( v_i-  v_j)=0$ for all indices~$j=1, ..., N-m$, which is equivalent to~$\chi_{ij}=0$ or~$v_j=v_i$. Thus, the non-zero entries of the vectors~$\chi_i$ indicate which nodes~$j$ have the same viral state as node~$i$. 

\begin{example}
Consider a network of~$N=5$ nodes with an invariant set~$\mathcal{V}$ of dimension $m=3$. Furthermore, consider that $\mathcal{V}_0=\emptyset$, which implies with (\ref{sdfasgrfasfasdfasdf}) that $\mathcal{V}=\mathcal{V}_{\neq 0}$. Thus, there are $N-m=2$ vectors~$\chi_4, \chi_5$. Suppose that the vectors $\chi_4, \chi_5$ are equal to $\chi_4 = (\chi_{41}, 0)^T$ and $\chi_5 = (0, \chi_{52})^T$, where $\chi_{41}, \chi_{52}\neq 0$. Then, (\ref{khbergoasrgstg}) implies that~$v_1 = v_4$ and~$v_2=v_5$ for every viral state $v\in\mathcal{V}$. Hence, the subspace~$\mathcal{V}=\operatorname{span}\{y_1, y_2, y_3\}$ is given by the basis vectors
\begin{align}
y_1 = \frac{1}{\sqrt{2}}\begin{pmatrix}
1 \\
0 \\
0 \\
1 \\
0
\end{pmatrix}, \quad y_2=\frac{1}{\sqrt{2}}\begin{pmatrix}
0 \\
1 \\
0 \\
0 \\
1
\end{pmatrix}, \quad y_3=\begin{pmatrix}
0 \\
0 \\
1 \\
0 \\
0
\end{pmatrix}.
\end{align}
For $l=1, 2, 3$, the eigenvector $x_{\phi(l)}$ of the infection rate matrix~$B$ equals a linear combination of the basis vectors $y_1, y_2, y_3$.
\end{example}
From (\ref{khbergoasrgstg}), we can determine disjoint subsets~$\mathcal{N}_1, \mathcal{N}_2, ... $ of the set of all nodes~$\mathcal{N}=\{1, ..., N\}$ as follows: If two nodes~$i,j$ are element of the same subset~$\mathcal{N}_l \subseteq \mathcal{N}$, then the viral states are equal, $v_i=v_j$, for every viral state~$v\in\mathcal{V}_{\neq 0}$. If a subset contains only one node, $\mathcal{N}_l=\{ i \}$,  then the viral state can be arbitrary~$v_i\in \mathbb{R}$, independently of the viral state~$v_j$ of other nodes $j\neq i$. Every subset defines a basis vector $y_l$ of the subspace~$\mathcal{V}_{\neq 0}$ as
\begin{align}\label{y_l_deff}
\left(y_l\right)_i = \begin{cases}
\frac{1}{\sqrt{|\mathcal{N}_l|}}\quad & \text{if} \quad i\in\mathcal{N}_l,\\
0\quad & \text{if} \quad i\not\in\mathcal{N}_l.\\
\end{cases}
\end{align}
 Then, the subspace~$\mathcal{V}_{\neq 0}$ equals the span of the vectors~$y_l$ of all subsets~$\mathcal{N}_l$. Since the dimension of the subspace~$\mathcal{V}_{\neq 0}$ is~$m_1$, there must be $m_1$ subsets~$\mathcal{N}_{1}, ...,\mathcal{N}_{m_1}$. Every node~$i$ is element of at most one subset~$\mathcal{N}_l$. Hence, the vectors~$y_l, y_{\tilde{l}}$ are orthogonal for $l\neq \tilde{l}$.

Furthermore, some nodes $i$ might not be element of any subset $\mathcal{N}_1, ..., \mathcal{N}_{m_1}$, which would imply that $(y_l)_i=0$ for all basis vectors $y_l$ of $\mathcal{V}_{\neq 0}$. We define the subset $\mathcal{N}_{m_1+1}$, whose elements are the nodes~$i$ that are not in any other subset~$\mathcal{N}_1, ..., \mathcal{N}_{m_1}$. As shown by Lemma~\ref{lemma:N_m1plus1_empty}, the set $\mathcal{N}_{m_1+1}$ is empty:
\begin{lemma}\label{lemma:N_m1plus1_empty}
Under Assumptions~\ref{assumption:delta_v_in_V} to \ref{assumption:matrix_B_symmetric}, it holds that~$\mathcal{N}_{m_1+1}=\emptyset$.
\end{lemma}
\begin{proof}
Under Assumption~\ref{assumption:v_positive}, there is a viral state vector~$v\in\mathcal{V}$ with positive entries. The positive viral state vector $v$ satisfies
\begin{align}\label{ikhbeargaergeg}
v = \sum^{m_1}_{l=1} z_l y_l + \sum^{m}_{l=m_1+1} z_l y_l
\end{align}
for some scalars $z_1, ..., z_m\in\mathbb{R}$. We denote the projection of the viral state~$v$ onto the subspace~$\mathcal{V}_0$ as
\begin{align}
v_\text{ker} = \sum^{m}_{l=m_1+1} z_l y_l
\end{align}
Every basis vector $y_l$ of the subspace $\mathcal{V}_{\neq 0}$ satisfies $(y_l)_i=0$ for all nodes $i\in\mathcal{N}_{m_1+1}$. Thus, we obtain with (\ref{ikhbeargaergeg}) that
\begin{align}\label{kjnakjbargareg}
\left(v_\text{ker}\right)_i=v_i>0
\end{align}
for all nodes $i\in\mathcal{N}_{m_1+1}$. Any vector~$\tilde{v}\in\mathcal{V}_{\neq 0}$ is orthogonal to the vector $v_\text{ker} \in\mathcal{V}_0$. Hence, it holds that
\begin{align}
\sum^N_{i=1} \left(\tilde{v}\right)_i \left(v_\text{ker}\right)_i = 0.
\end{align}
We split the sum
\begin{align}
\sum^{m_1}_{l=1}\sum_{i\in\mathcal{N}_l} \left(\tilde{v}\right)_i \left(v_\text{ker}\right)_i +\sum_{i\in\mathcal{N}_{m_1+1}} \left(\tilde{v}\right)_i \left(v_\text{ker}\right)_i = 0.
\end{align}
Since $\left(\tilde{v}\right)_i=0$ for all nodes~$i\in\mathcal{N}_{m_1+1}$, we obtain that
\begin{align}\label{aergergregrgr}
\sum^{m_1}_{l=1}\sum_{i\in\mathcal{N}_l} \left(\tilde{v}\right)_i \left(v_\text{ker}\right)_i = 0\quad \forall \tilde{v}\in\mathcal{V}_{\neq 0}.
\end{align}
Furthermore, we define the $N\times 1$ vector $u_a$ with the entries
\begin{align}
(u_a)_i=\begin{cases}
1 \quad &\text{if}~i\not\in\mathcal{N}_{m_1+1}, \\
0 &\text{if}~i\in\mathcal{N}_{m_1+1}.
\end{cases}
\end{align}
From the definition of the basis vectors $y_l$ in (\ref{y_l_deff}), it follows that the vector $u_a$ equals 
\begin{align}
u_a = \sum^{m_1}_{l=1} \sqrt{|\mathcal{N}_l|} y_l .
\end{align}
Thus, vector $u_a$ is element of $\mathcal{V}_{\neq 0}$. Since the vector $v_\text{ker}$ is in the kernel of the matrix~$B$, it holds that $B v_\text{ker}=0$, which implies that
\begin{align}\label{nregaregaerg}
u^T_a B v_\text{ker}=0.
\end{align}
We decompose the vector $v_\text{ker}$ as $v_\text{ker} = v_{\text{ker},a} + v_{\text{ker},b}$, where the first addend equals
\begin{align}
\left(v_{\text{ker},a}\right)_i = \begin{cases}
\left(v_\text{ker}\right)_i \quad &\text{if}~i\not\in\mathcal{N}_{m_1+1}, \\
0 &\text{if}~i\in\mathcal{N}_{m_1+1},
\end{cases}
\end{align}
and the second addend equals
\begin{align}\label{ikjnreggerg}
\left(v_{\text{ker},b}\right)_i=\begin{cases}
0 \quad &\text{if}~i\not\in\mathcal{N}_{m_1+1}\\
\left(v_\text{ker}\right)_i  &\text{if}~i\in\mathcal{N}_{m_1+1}.
\end{cases}
\end{align}
Then, (\ref{nregaregaerg}) becomes
\begin{align}
u^T_a B v_{\text{ker},a} +u^T_a B v_{\text{ker},b}=0.
\end{align}
Since $u_a\in \mathcal{V}_{\neq 0}$ and $\mathcal{V}_{\neq 0}$ is an invariant subspace of the matrix $B$, it holds that $B u_a\in \mathcal{V}_{\neq 0}$. Thus, (\ref{aergergregrgr}) implies that $u^T_a B v_{\text{ker},a}=0$, and we obtain that
\begin{align}
u^T_a B v_{\text{ker},b}=0,
\end{align}
which is equivalent to
\begin{align}
\sum^{m_1}_{l=1}\sum_{i\in\mathcal{N}_l} \sum^N_{j=1}\beta_{ij}\left(v_{\text{ker},b}\right)_j=0.
\end{align}
With the definition of the vector $v_{\text{ker},b}$ in (\ref{ikjnreggerg}), we obtain that
\begin{align}\label{aergergaergregwqqq}
\sum^{m_1}_{l=1}\sum_{i\in\mathcal{N}_l} \sum_{j\in\mathcal{N}_{m_1+1}}\beta_{ij}\left(v_\text{ker}\right)_j=0.
\end{align}
As stated by (\ref{kjnakjbargareg}), the entries~$\left(v_\text{ker}\right)_j$ are positive for all nodes $j\in\mathcal{N}_{m_1+1}$. Furthermore, the infection rates $\beta_{ij}$ are non-negative under Assumption~\ref{assumption:spreading_rates}. Hence, (\ref{aergergaergregwqqq}) is satisfied only if $\beta_{ij}=0$ for all nodes $j\in\mathcal{N}_{m_1+1}$ and $i\in\mathcal{N}_l$ for all subsets $l=1, ..., m_1$. In other words, the nodes in $\mathcal{N}_{m_1+1}$ are not connected to any nodes in $\mathcal{N}_1, ..., \mathcal{N}_{m_1}$, which contradicts the irreducibility of the matrix $B$ under Assumption~\ref{assumption:matrix_B_symmetric}. Hence, it must hold that $\mathcal{N}_{m_1+1}=\emptyset$.
\end{proof}
 Since $\mathcal{N}_{m_1+1}=\emptyset$, it holds that $\mathcal{N}_1 \cup ...\cup \mathcal{N}_{m_1}=\mathcal{N}$. Hence, the disjoint subsets~$\mathcal{N}_1$, ..., $\mathcal{N}_{m_1}$ define a partition of the set of all nodes~$\mathcal{N}=\{1, ..., N\}$. To complete the proof of Theorem~\ref{theorem:invariant_sets_are_equitable_partitions}, we must show that the subsets $\mathcal{N}_{1}$, ..., $\mathcal{N}_{m_1}$ are an \textit{equitable} partition of the infection rate matrix~$B$. Hence, we must show that the sum of the infection rates~$\beta_{ij}$, 
 \begin{align}\label{d_ij_khjsdbfsad}
\sum_{ j \in \mathcal{N}_l} \beta_{ij},
\end{align}
is the same for all nodes $i\in\mathcal{N}_p$ and all cells $l,p=1,...,m_1$. Lemma~\ref{lemma:invariance} states that
\begin{align}
\mathcal{V}_{\neq 0}&=\operatorname{span}\left\{y_{1}, ..., y_{m_1}\right\}\\
&=\operatorname{span}\left\{x_{\phi(1)}, ..., x_{\phi(m_1)}\right\}.
\end{align}
Thus, there must be some nonsingular $m_1\times m_1$ matrix $H$ such that
\begin{align}\label{kujarfgrdgerg}
\begin{pmatrix}
x_{\phi(1)} &... & x_{\phi(m_1)}
\end{pmatrix}
= \begin{pmatrix}
y_{1} & ... & y_{m_1}
\end{pmatrix} H.
\end{align}
Since the set eigenvectors~$x_i$ and the set of vectors~$y_l$ are orthonormal, the matrix~$H$ is orthogonal\footnote{Since $x^T_i x_j=1$ if $i=j$ and $x^T_i x_j=0$ if $i\neq j$ and analogously for the vectors~$y_i,y_j$, it follows from $x^T_i x_j=y^T_i H^T H y_j$ that the matrix~$H$ is orthogonal.}. The eigendecomposition of the matrix~$B$ reads
\begin{align}
B = &\begin{pmatrix}
x_{\phi(1)} &... & x_{\phi(m_1)}
\end{pmatrix} \operatorname{diag}\left(\lambda_{\phi(1)}, ..., \lambda_{\phi(m_1)} \right)
\begin{pmatrix}
x^T_{\phi(1)} \\
\vdots\\
x^T_{\phi(m_1)}
\end{pmatrix} \\
&+\begin{pmatrix}
x_{\phi(m_1+1)} &... & x_{\phi(m)}
\end{pmatrix} \operatorname{diag}\left(\lambda_{\phi(m_1+1)}, ..., \lambda_{\phi(m)} \right)
\begin{pmatrix}
x^T_{\phi(m_1+1)} \\
\vdots\\
x^T_{\phi(m)}
\end{pmatrix} \\
&+\begin{pmatrix}
x_{\phi(m+1)} &... & x_{\phi(N)}
\end{pmatrix} \operatorname{diag}\left(\lambda_{\phi(m+1)}, ..., \lambda_{\phi(N)} \right)
\begin{pmatrix}
x^T_{\phi(m+1)} \\
\vdots\\
x^T_{\phi(N)}
\end{pmatrix}.
\end{align}
With (\ref{kujarfgrdgerg}), and since the eigenvalues~$\lambda_{\phi(l)}=0$ for $l=m_1+1, ..., m$, we obtain that
\begin{align} \label{sdfawfrfaderfgrer}
B =&\begin{pmatrix}
y_{1} & ... & y_{m_1}
\end{pmatrix} H \operatorname{diag}\left(\lambda_{\phi(1)}, ..., \lambda_{\phi(m_1)} \right)
H^T
\begin{pmatrix}
y^T_{1} \\
\vdots\\
 y^T_{m_1}
\end{pmatrix} \\
&+\begin{pmatrix}
x_{\phi(m+1)} &... & x_{\phi(N)}
\end{pmatrix} \operatorname{diag}\left(\lambda_{\phi(m+1)}, ..., \lambda_{\phi(N)} \right)
\begin{pmatrix}
x^T_{\phi(m+1)} \\
\vdots\\
x^T_{\phi(N)}
\end{pmatrix}.
\end{align}
Consider two nodes $i\in \mathcal{N}_p$ and a subset $\mathcal{N}_l$ for some $l=1, ..., m_1$. Since 
\begin{align}
(y_l)_j = \begin{cases}
\frac{1}{\sqrt{|\mathcal{N}_l|}}\quad &\text{if}~j\in\mathcal{N}_l, \\
0 &\text{if}~j\not\in\mathcal{N}_l,
\end{cases}
\end{align} 
we can express the sum (\ref{d_ij_khjsdbfsad}) as
\begin{align}
\sum_{ j \in \mathcal{N}_l} \beta_{ij} = \sqrt{|\mathcal{N}_l|}  \begin{pmatrix}
\beta_{i1} & ... & \beta_{iN}
\end{pmatrix} y_l.
\end{align}
Thus, with the $N\times 1$ basic vector~$e_i$, it holds that
\begin{align}
\sum_{ j \in \mathcal{N}_l} \beta_{ij} = \sqrt{|\mathcal{N}_l|}  e^T_i B y_l.
\end{align}
From the orthogonality of the vectors $y_1, ..., y_{m_1}$ and from $x^T_{\phi(k)} y_l=0$ for $k=m+1, ..., N$, we obtain with (\ref{sdfawfrfaderfgrer}) that
\begin{align}\label{kjaergaregaregre}
\sum_{ j \in \mathcal{N}_l} \beta_{ij} = \sqrt{|\mathcal{N}_l|}  e^T_i \begin{pmatrix}
y_{1} & ... & y_{m_1}
\end{pmatrix} H \operatorname{diag}\left(\lambda_{\phi(1)}, ..., \lambda_{\phi(m_1)} \right)
H^T e_{m_1\times 1,l},
\end{align}
where the $l$-th entry of the $m_1\times 1$ vector $e_{m_1\times 1,l}$ equals one, and the other entries of $e_{m_1\times 1,l}$ equal zero. Since node $i$ is element of exactly one subset~$\mathcal{N}_p$, it holds that 
\begin{align}
e^T_i \begin{pmatrix}
y_{1} & ... & y_{m_1}
\end{pmatrix} = \frac{1}{\sqrt{|\mathcal{N}_p|}} \tilde{e}^T_{m_1\times 1,p}.
\end{align}
Then, (\ref{kjaergaregaregre}) becomes
\begin{align}
\sum_{ j \in \mathcal{N}_l} \beta_{ij} = d_{il},
\end{align}
where 
\begin{align}
d_{il}=\frac{\sqrt{|\mathcal{N}_l|}}{\sqrt{|\mathcal{N}_p|}} e^T_{m_1\times 1,p} H \operatorname{diag}\left(\lambda_{\phi(1)}, ..., \lambda_{\phi(m_1)} \right)
H^T e_{m_1\times 1,l}
\end{align}
is the same for all nodes $i\in\mathcal{N}_p$, which completes the proof.

\section{Proof of Corollary~\ref{corollary:one_dimensional_solution}} \label{appendix:one_dimensional_solution}

Since $R_0>1$, the viral state $v(t)$ converges to a positive steady state $v_\infty$ as $t\rightarrow \infty$. Thus, the steady state $v_\infty$ must be element of the $m=1$ dimensional invariant set $\mathcal{V}=\operatorname{span}\{y_1\}$, which implies that $v_\infty=\tilde{c}y_1$ for some scalar $c$. Hence, the unit-length agitation mode equals either $y_1=v_\infty/\lVert v_\infty\rVert_2$ or $y_1=-v_\infty/\lVert v_\infty\rVert_2$. Without loss of generality assume that $y_1=v_\infty/\lVert v_\infty\rVert_2$. Then, under Assumption~\ref{assumption:matrix_B_symmetric}, the matrix $B$ is connected, which implies that $By_1\neq 0$ since the vector $y_1$ is positive. Thus, the subspace $\mathcal{V}_0$ must be empty.

To prove Corollary~\ref{corollary:one_dimensional_solution}, we must show two directions. \textbf{\quotes{If} direction}: Suppose the infection rate matrix $B$ is regular. Then, the viral state $v_{\infty,i}$ is the same for all nodes $i$, and $v(0)\in\mathcal{V}$ implies that $v_i(0)=v_j(0)$ for all nodes $i,j$. Since the matrix $B$ is regular and the initial viral state $v_i(0)$ is the same for every node $i$, the approximation $v_\textrm{apx}(t)=c(t)v_\infty$ is exact \cite{van2014sis,prasse2019time}. Since $v(t)=c(t)v_\infty$ at every time $t$, the invariant set $\mathcal{V}=\operatorname{span}\{y_1\}$ is indeed a one-dimensional invariant set of NIMFA.

\textbf{\quotes{Only if} direction}: Suppose the one-dimensional subspace $\mathcal{V}=\operatorname{span}\{y_1\}$ is an invariant set of NIMFA. Then, Theorem~\ref{theorem:invariant_sets_are_equitable_partitions} yields that the infection rate matrix $B$ has the equitable partition $\pi=\{\mathcal{N}_1\}$, where the cell $\mathcal{N}_1=\{1, ..., N\}$ contains all nodes. Thus, (\ref{khbrffdg}) yields, that there exists some degree $d_{11}$ which satisfies
\begin{align}
d_{11} &=\sum_{k \in \mathcal{N}_1} \beta_{ik} \\
&=\sum^N_{k=1} \beta_{ik}
\end{align}
for all nodes $i$. Thus, we obtain with definition (\ref{asdafsdfseddasss}) that the matrix $B$ is regular.

\section{Proof of Theorem~\ref{theorem:decompose_dynamics}}\label{appendix:decompose_dynamics}
By assumption, the infection rates $\beta_{i,j}$ are the same for all nodes $i$ in any cell $\mathcal{N}_l$ and all nodes $j$ in any cell $\mathcal{N}_p$. Thus, with the definition of the vectors $y_1$, ..., $y_r$ in (\ref{y_l_definition}), the symmetric infection rate matrix equals
\begin{align}\label{ljnlkjnlregergreg}
B =\begin{pmatrix}
y_1 & ... & y_r
\end{pmatrix}  \tilde{B}_{\mathcal{V}_{\neq 0}} \begin{pmatrix}
y^T_1 \\
\vdots \\
y^T_r
\end{pmatrix}
\end{align}
for some symmetric $r\times r$ matrix $\tilde{B}_{\mathcal{V}_{\neq 0}}$. Since the kernel $\operatorname{ker}(B)$ is the orthogonal complement of the subspace $\mathcal{V}_{\neq 0}$, it holds that $\mathbb{R}^N = \mathcal{V}_{\neq 0} \oplus \operatorname{ker}(B)$. Thus, any viral state vector $v(t)\in[0,1]^N$ can be decomposed as $v(t)= \tilde{v}(t) + v_\text{ker}(t)$, where $\tilde{v}(t)\in \mathcal{V}_{\neq 0}$ and $v_\text{ker}(t)\in \operatorname{ker}(B)$. With the decomposition $v(t)= \tilde{v}(t) + v_\text{ker}(t)$, NIMFA (\ref{NIMFA_stacked}) becomes
\begin{align}
	\frac{d v (t)}{d t } & = - S \left( \tilde{v}(t) + v_\text{ker}(t) \right) + \textup{\textrm{diag}}\left( u - \tilde{v}(t) - v_\text{ker}(t)\right) B \left( \tilde{v}(t) + v_\text{ker}(t) \right) \\
	&= - S \tilde{v}(t) - S v_\text{ker}(t) + \textup{\textrm{diag}}\left( u - \tilde{v}(t) - v_\text{ker}(t)\right) B \tilde{v}(t),
\end{align}
where the second equality follows from $Bv_\text{ker}(t)=0$. Further rearrangement yields that
\begin{align} \label{kjnkjsferfvvvvvv}
	\frac{d v (t)}{d t } & = \left(  B - S \right) \tilde{v}(t) - \textup{\textrm{diag}}\left( \tilde{v}(t)\right) B \tilde{v}(t) - S v_\text{ker}(t) - \textup{\textrm{diag}}\left( v_\text{ker}(t)\right) B \tilde{v}(t).
\end{align}

We decompose the derivative $dv(t)/dt$ into two addends, by making use of two lemmas:
\begin{lemma} \label{lemma:first_addend_orth}
Suppose that the assumptions in Theorem~\ref{theorem:decompose_dynamics} hold true. Then, if $\tilde{v}\in\mathcal{V}_{\neq 0}$, the vector
\begin{align} \label{kljnkjnkjnasasssssaaa}
B \tilde{v} - S \tilde{v} - \textup{\textrm{diag}}\left( \tilde{v}\right) B \tilde{v}
\end{align}
is element of $\mathcal{V}_{\neq 0}$.
\end{lemma}
\begin{proof}
We consider the three addends of the vector (\ref{kljnkjnkjnasasssssaaa}) separately. First, (\ref{ljnlkjnlregergreg}) shows that the addend $B\tilde{v}$ is element of $\mathcal{V}_{\neq 0}$ if $\tilde{v}\in\mathcal{V}_{\neq 0}$. Second, we consider the addend $S \tilde{v}$. By assumption, the curing rates $\delta_i$ are the same for all nodes $i$ in the same cell $\mathcal{N}_l$. Thus, we obtain from the definition of the agitation modes $y_l$ in (\ref{y_l_definition}) that 
\begin{align} \label{ljnlrgergetg}
S y_l = \delta_i y_l
\end{align}
for $l=1, ..., r$, where $i$ denotes an arbitrary node in cell $\mathcal{N}_l$. Since the agitation modes $y_1$, ..., $y_r$ span the subspace $\mathcal{V}_{\neq 0}$, (\ref{ljnlrgergetg}) implies that $S\tilde{v}$ if $\tilde{v}\in\mathcal{V}_{\neq 0}$.

Third, we consider the addend $\textup{\textrm{diag}}\left( \tilde{v}\right) B \tilde{v}$. Since $\tilde{v}\in\mathcal{V}_{\neq 0}$, it holds that
\begin{align}
\tilde{v} = \sum^r_{l=1} \left(y^T_l \tilde{v} \right) y_l.
\end{align}
Similarly, since $B \tilde{v}\in\mathcal{V}_{\neq 0}$, it holds that
\begin{align} \label{kjnkjnaaasssscccc}
B \tilde{v}= \sum^r_{l=1} \left(y^T_l B \tilde{v} \right) y_l.
\end{align}
Thus, we obtain that
\begin{align}\label{lojnlreeeee}
\textup{\textrm{diag}}\left( \tilde{v}\right) B \tilde{v}= \sum^r_{l=1} \sum^r_{p=1} \left(y^T_l \tilde{v} \right)\left(y^T_p B \tilde{v} \right) \operatorname{diag}(y_l) y_p.
\end{align}
From the definition of the vectors $y_l$ in (\ref{y_l_definition}) it follows that
\begin{align} \label{sdfgdrgegteg}
\textup{\textrm{diag}}\left( y_l \right) y_p = \begin{cases}
y^2_l \quad &\text{if} \quad l = p,\\
0 &\text{if} \quad l \neq p,
\end{cases}
\end{align}
where the $N\times 1$ vector $y^2_l = \left( ( y_l )^2_1, ..., ( y_l )^2_N \right)^T$ denotes Hadamard product of the vector $y_l$ with itself. Thus, (\ref{lojnlreeeee}) becomes
\begin{align} \label{ljknljnergergerg}
\textup{\textrm{diag}}\left( \tilde{v}\right) B \tilde{v}=\sum^r_{l=1}  \left(y^T_l \tilde{v} \right)\left(y^T_l B \tilde{v} \right) y^2_l.   
\end{align}
With (\ref{y_l_definition}), the Hadamard product $y^2_l$ equals 
\begin{align}
( y_l )^2_i = \begin{cases}
\frac{1}{\left|\mathcal{N}_l\right|} \quad &\text{if} \quad i \in\mathcal{N}_l,\\
0 &\text{if} \quad i \not\in\mathcal{N}_l,
\end{cases}
\end{align}
which implies that $( y_l )^2=y_l/\sqrt{\left|\mathcal{N}_l\right|}$ and yields with (\ref{ljknljnergergerg}) that
\begin{align} 
\textup{\textrm{diag}}\left( \tilde{v}\right) B \tilde{v}= \sum^r_{l=1}  \frac{\left(y^T_l \tilde{v} \right)\left(y^T_l B \tilde{v} \right)}{\sqrt{\left|\mathcal{N}_l\right|}} y_l.   
\end{align}
Thus, the vector $\textup{\textrm{diag}}\left( \tilde{v}\right) B \tilde{v}$ is a linear combination of the vectors $y_1$, ..., $y_r$, which implies that $\textup{\textrm{diag}}\left( \tilde{v}\right) B \tilde{v} \mathcal{V}_{\neq 0}$. Hence, we have shown that all three addends of the vector (\ref{kljnkjnkjnasasssssaaa}) are in $\mathcal{V}_{\neq 0}$, which completes the proof.
\end{proof}

\begin{lemma}\label{lemma:second_addend_orth}
Suppose that the assumptions in Theorem~\ref{theorem:decompose_dynamics} hold true. Then, if $\tilde{v}\in\mathcal{V}_{\neq 0}$ and $v_\textup{ker}\in\operatorname{ker}(B)$, the vector
\begin{align} \label{kjbjkhbeeeaass}
S v_\textup{ker} + \operatorname{diag}\left( v_\textup{ker}\right) B \tilde{v}
\end{align}
is element of $\operatorname{ker}(B)$.
\end{lemma}
\begin{proof}
The kernel $\operatorname{ker}(B)$ is the orthogonal complement of the subspace $\mathcal{V}_{\neq 0}$. Thus, the vector (\ref{kjbjkhbeeeaass}) is element of $\operatorname{ker}(B)$ if $S v_\text{ker}$ is orthogonal to every basis vector $y_1$, ..., $y_r$ of the subspace $\mathcal{V}_{\neq 0}$. We show separately that both addends of the vector (\ref{kjbjkhbeeeaass}) are orthogonal to every vector $y_1$, ..., $y_r$. First, for any $l=1, ..., r$, we obtain for the first addend in (\ref{kjbjkhbeeeaass}) that 
\begin{align}
y^T_l S v_\text{ker} = \left( S y_l \right)^T v_\text{ker},
\end{align}
since the matrix $S$ is symmetric. With (\ref{ljnlrgergetg}), we obtain for an arbitrary node $i\in\mathcal{N}_l$ that 
\begin{align}
y^T_l S v_\text{ker} = \delta_i y^T_l v_\text{ker} =0.
\end{align}
Thus, the addend $S v_\text{ker}$ is element of $\operatorname{ker}(B)$. 

Second, for any $l=1, ..., r$, we obtain for the second addend in (\ref{kjbjkhbeeeaass}) with (\ref{kjnkjnaaasssscccc}) that 
\begin{align} 
y^T_l \textup{\textrm{diag}}\left( v_\textup{ker}\right) B \tilde{v} &=  \sum^{r}_{q=1}\left( y^T_l B\tilde{v} \right) y^T_l \textup{\textrm{diag}}\left( v_\textup{ker}\right) y_q \\
&=\sum^{r}_{q=1}\left( y^T_q B\tilde{v} \right) v^T_\textup{ker} \textup{\textrm{diag}}\left( y_l \right) y_q.
\end{align}
Analogous steps as in the proof of Lemma~\ref{lemma:first_addend_orth} yield that
\begin{align} 
y^T_l \textup{\textrm{diag}}\left( v_\textup{ker}\right) B \tilde{v} &= \frac{\left( y^T_l B\tilde{v} \right)}{\sqrt{\left|\mathcal{N}_l\right|}} v^T_\textup{ker}  y_l.
\end{align}
Thus, by the orthogonality of the vectors $v_\textup{ker}$ and $y_l$,
\begin{align}
y^T_l \textup{\textrm{diag}}\left( v_\textup{ker}\right) B \tilde{v} =0,
\end{align}
which completes the proof.
\end{proof}

With Lemma~\ref{lemma:first_addend_orth} and Lemma~\ref{lemma:second_addend_orth}, we obtain from (\ref{kjnkjsferfvvvvvv}) that 
\begin{align} 
	\frac{d v (t)}{d t } & = 	\frac{d \tilde{v}(t)}{d t } + 	\frac{d v_\textup{ker}(t)}{d t },	
\end{align}
where
\begin{align} 
	\frac{d \tilde{v}(t)}{d t } = - S \tilde{v}(t) + \textup{\textrm{diag}}\left(  u - \tilde{v}(t)\right) B \tilde{v}(t)
\end{align}
	and
\begin{align} 
	\frac{d v_\textup{ker}(t)}{d t } = - S v_\text{ker}(t) - \operatorname{diag}\left( v_\text{ker}(t)\right) B \tilde{v}(t),
\end{align}
	which completes the proof, since 
	\begin{align}
	 \operatorname{diag}\left( v_\text{ker}(t)\right) B \tilde{v}(t) =  \operatorname{diag}\left( B \tilde{v}(t)\right)  v_\text{ker}(t).
	\end{align}

\section{Proof of Theorem~\ref{theorem:complete_graph_solution}} \label{appendix:complete_graph_solution}

Since the spreading rates are homogeneous, $\beta_{ij}= \beta$ and $\delta_i = \delta$, the infection rate matrix equals
\begin{align}\label{kjnkjjjjssss}
B = \beta u u^T,
\end{align}
and the curing rate matrix equals
\begin{align}\label{uhiubiaaaassscc}
S = \delta I.
\end{align}
Thus, with $r=1$ cell $\mathcal{N}_1=\{1, ..., N\}$, Theorem~\ref{theorem:decompose_dynamics} yields that the viral state $v(t)$ can be decomposed as $v(t)=\tilde{v}(t) + v_\text{ker}(t)$. We prove Theorem~\ref{theorem:complete_graph_solution} in two steps. First, we show that the projection $\tilde{v}(t)$ equals $c_1(t)v_\infty$ at every time $t$. Second, we prove that the projection $v_\text{ker}(t)$ equals $c_2(t) y_2$ at every time $t$.

\subsection{Projection on the subspace $\mathcal{V}_{\neq 0}$}\label{khjbkbkjrfwefwfwefaa}
With the reduced-size curing rate matrix $S^\pi = \delta$ and the quotient matrix $B^{\pi} = N \beta$, Theorem~\ref{theorem:equitableOriginal} yields that the projection on the subspace $\mathcal{V}_{\neq 0}$ satisfies $\tilde{v}(t)=v^{\pi}(t) u$. The evolution (\ref{kjnkjnertgverg}) of the reduced-size, scalar viral state $v^{\pi}(t)$ becomes
\begin{align}
	\frac{d v^{\pi}(t)}{d t } & = - \delta  v^{\pi}(t) + \left( 1 - v^{\pi}(t)\right) N \beta v^{\pi}(t),
\end{align}
whose solution equals \cite{van2014sis,prasse2019time}
\begin{align}
v^{\pi}(t) = \frac{v^{\pi}_\infty}{2}\left( 1 + \tanh\left( \frac{w}{2} t +\Upsilon_1(0)  \right) \right)
\end{align}
with the reduced-size steady state $v^{\pi}_\infty=1-\frac{\delta}{\beta N}$, the viral slope $w=\beta N - \delta$ and the constant
\begin{align}
\Upsilon_1(0) = \operatorname{artanh}\left( 2 \frac{v(0)}{v_\infty} - 1\right).
\end{align}
Thus, the projection $\tilde{v}(t)=v^{\pi}(t) u$ is equal to $c_1(t) v_\infty$ at every time $t$.

\subsection{Projection on the kernel $\operatorname{ker}(B)$}

With (\ref{kjnkjjjjssss}) and (\ref{uhiubiaaaassscc}), Theorem~\ref{theorem:decompose_dynamics} yields that the projection $v_\text{ker}(t)$ obeys
\begin{align}
\frac{d v_\text{ker}(t)}{dt} = -\left( \delta I + \beta \operatorname{diag}\left( u u^T \tilde{v}(t)\right) \right)v_\textup{ker}(t).
\end{align}
Since $\tilde{v}(t)=c_1(t) v_\infty$ and $v_\infty = v_{\infty, i} u$ for an arbitrary node $i$, we obtain that
\begin{align}
\frac{d v_\text{ker}(t)}{dt} &= -\left( \delta I + \beta N c_1(t)  v_{\infty, i} I \right)v_\textup{ker}(t).
\end{align}
From the function $c_1(t)$ given in (\ref{c_1_closedform}), it follows that
\begin{align} \label{ojnnasssdasdasdsd}
\frac{d v_\text{ker}(t)}{dt} &= - \delta v_\textup{ker}(t) - \frac{\beta N v_{\infty, i}}{2}\left( 1 + \tanh\left( \frac{w}{2} t +\Upsilon_1(0)  \right) \right) v_\textup{ker}(t).
\end{align}
For any initial condition $v_\text{ker}(0)\in \operatorname{ker}(B)$, the right side of (\ref{ojnnasssdasdasdsd}) is element of the one-dimensional subspace $\operatorname{span}\{v_\text{ker}(0)\}$. Thus, the projection $v_\textup{ker}(t)$ obeys $v_\textup{ker}(t)=c_2(t) v_\textup{ker}(0)$. We solve (\ref{ojnnasssdasdasdsd}) in two steps. First, we compute the initial condition $v_\text{ker}(0)$. Since $v(0) = v_\text{ker}(0) + c_1(0) v_\infty $, the initial condition $v_\text{ker}(0)$ is obtained as
\begin{align}
v_\text{ker}(0) &= v(0) - c_1(0) v_\infty \\
&= v(0) - \frac{v^T_\infty v(0)}{\lVert v_\infty \rVert^2_2} v_\infty.
\end{align}
Since $v_\infty = v_{\infty, i}u$, it follows that
\begin{align}
v_\text{ker}(0) &= v(0) - \frac{1}{N}u^T v(0) u,
\end{align}
which simplifies to 
\begin{align}
v_\text{ker}(0) &= \left( I - \frac{1}{N}u u^T\right) v(0).
\end{align}
Second, using $v_\textup{ker}(t)=c_2(t) v_\textup{ker}(0)$, we project (\ref{ojnnasssdasdasdsd}) on the initial condition $v_\textup{ker}(0)$ to obtain that the scalar function $c_2(t)$ obeys the linear differential equation 
\begin{align}  \label{ljnlkjnlkjnsfgsv}
\frac{d c_2(t)}{dt} &= - \delta c_2(t) - \frac{\beta N v_{\infty, i}}{2}\left( 1 + \tanh\left( \frac{w}{2} t +\Upsilon_1(0)  \right) \right) c_2(t) 
\end{align}
and hence, with the constant $\Phi = \beta N v_{\infty, i}/2 +\delta$,
\begin{align}  
\log\left( c_2(t) \right) &= - \int^t_0 \left( \Phi + \frac{1}{2}\beta N v_{\infty, i}\tanh\left( \frac{w}{2} \xi +\Upsilon_1(0)  \right) \right) d\xi. 
\end{align}
The integral of the hyperbolic tangent equals to the logarithm of the hyperbolic cosine \cite{abramowitz1965handbook},
\begin{align}
\int \tanh\left( \xi \right) d\xi = \log\left( \cosh(\xi) \right),
\end{align}
which yields that 
\begin{align}  \label{kjbnkjnkasasddfasdas}
\log\left( c_2(t) \right) &= - \Phi t -  \frac{\beta N v_{\infty, i}}{2} \frac{2}{w} \log\left( \cosh\left( \frac{w}{2} t +\Upsilon_1(0)  \right)\right) +K(0) 
\end{align}
for some constant $K(0)$. With the definition of the viral slope $w$ in Subsection~\ref{khjbkbkjrfwefwfwefaa} and $v_{\infty,i}=1 - \frac{\delta}{\beta N}$, we obtain that 
\begin{align}
\frac{\beta N v_{\infty, i}}{w} = \frac{\beta N (1 - \frac{\delta}{\beta N})  }{\beta N -\delta} =1.
\end{align}
Thus, (\ref{kjbnkjnkasasddfasdas}) becomes
\begin{align}  
\log\left( c_2(t) \right) &= - \Phi t + \log\left( \cosh\left( \frac{w}{2} t +\Upsilon_1(0)  \right)^{-1}\right) +K(0),
\end{align}
and we obtain, with the hyperbolic secant $\operatorname{sech}(x)=\cosh(x)^{-1}$, that 
\begin{align}  \label{kljnkjnsdfsdfgsdfgagfsd}
c_2(t) &= \Upsilon_2(0) e^{ - \Phi t } \operatorname{sech}\left( \frac{w}{2} t +\Upsilon_1(0)  \right).
\end{align}
with the constant $\Upsilon_2(0)=\exp(K(0))$. At the initial time $t=0$, (\ref{kljnkjnsdfsdfgsdfgagfsd}) becomes
\begin{align}  
c_2(0) &= \Upsilon_2(0) \operatorname{sech}\left( \Upsilon_1(0)  \right),
\end{align}
and it holds that 
\begin{align}
c_2(0) = \frac{v^T_\text{ker}(0) v(0)}{\lVert v_\text{ker}(0) \rVert^2_2}.
\end{align}
Thus, with $\operatorname{sech}(x)=\cosh(x)^{-1}$, we obtain the constant $\Upsilon_2(0)$ as (\ref{kjbnkjnsdfaserwfaerfaedrfaerf}), which completes the proof.

\section{Proof of Theorem~\ref{theorem:NIMFA_bounds}} 
\label{appendix:NIMFA_bounds}
The viral state $\tilde{v}_i(t)$ evolves as
\begin{align}
\frac{d \tilde{v}_i(t)}{dt} &= \tilde{f}_{\text{NIMFA}, i}( \tilde{v}(t) ),
\end{align}
where we define, for every node $i$, 
\begin{align}\label{khjbnksdsdfsdsdssdd}
\tilde{f}_{\text{NIMFA}, i}( \tilde{v}(t) ) = -\tilde{\delta}_i \tilde{v}_i(t) + \left(1-\tilde{v}_i(t)\right)\sum^N_{j=1} \tilde{\beta}_{ij} \tilde{v}_j(t) .
\end{align}
Since $\tilde{\beta}_{ij}\ge \beta_{ij}$ and $\tilde{\delta}_i\le \delta_i$ for all nodes $i$, we obtain an upper bound on NIMFA (\ref{NIMFA_continuous}) as
\begin{align}
\frac{d v_i(t)}{dt} &\le -\tilde{\delta}_i v_i(t) + \left(1-v_i(t)\right)\sum^N_{j=1} \tilde{\beta}_{ij} v_j(t) \\
&= \tilde{f}_{\text{NIMFA}, i}(v(t)).
\end{align}
Since $d v_i(t)/dt \le \tilde{f}_{\text{NIMFA}, i}(v(t))$, we can apply the \textit{Kamke-M{\"u}ller condition} \cite{kamke1932theorie, muller1927fundamentaltheorem}, see also \cite{kiss2017mathematics}: If $v\le \tilde{v}$ and $v_i=\tilde{v}_i$ implies that $\tilde{f}_{\text{NIMFA}, i}(v)\le \tilde{f}_{\text{NIMFA}, i}\left(\tilde{v}\right)$ for all nodes $i$, then $v(0)\le \tilde{v}(0)$ implies that $v(t)\le \tilde{v}(t)$ at every time $t\ge 0$.

Thus, it remains to show that $v\le \tilde{v}$ and $v_i=\tilde{v}_i$ implies that $\tilde{f}_{\text{NIMFA}, i}(v)\le \tilde{f}_{\text{NIMFA}, i}\left(\tilde{v}\right)$. From (\ref{khjbnksdsdfsdsdssdd}), we obtain that
\begin{align} 
\tilde{f}_{\text{NIMFA}, i}(v)- \tilde{f}_{\text{NIMFA}, i}\left(\tilde{v}\right) = -\tilde{\delta}_i \left(v_i - \tilde{v}_i\right) + \left(1-v_i\right)\sum^N_{j=1} \tilde{\beta}_{ij} v_j - \left(1-\tilde{v}_i\right)\sum^N_{j=1} \tilde{\beta}_{ij} \tilde{v}_j.
\end{align}
From $v_i=\tilde{v}_i$, it follows that
\begin{align} 
\tilde{f}_{\text{NIMFA}, i}(v)- \tilde{f}_{\text{NIMFA}, i}\left(\tilde{v}\right) = \left(1-v_i\right)\sum^N_{j=1} \tilde{\beta}_{ij} v_j - \left(1-v_i\right)\sum^N_{j=1} \tilde{\beta}_{ij} \tilde{v}_j,
\end{align}
which yields that 
\begin{align} 
\tilde{f}_{\text{NIMFA}, i}(v)- \tilde{f}_{\text{NIMFA}, i}\left(\tilde{v}\right) &= \sum^N_{j=1} \tilde{\beta}_{ij}\left( v_j -v_i v_j -\tilde{v}_j + v_i \tilde{v}_j \right) \\
&=\sum^N_{j=1} \tilde{\beta}_{ij}\left(1-v_i \right) \left( v_j -\tilde{v}_j\right).
\end{align}
Since $\left( v_j -\tilde{v}_j\right)\le 0$, we obtain that $\tilde{f}_{\text{NIMFA}, i}(v)\le \tilde{f}_{\text{NIMFA}, i}\left(\tilde{v}\right)$, which completes the proof.

\section{Proof of Theorem \ref{theorem:bound_by_equitable}}
\label{appendix:bound_by_equitable}
Here, we prove that $v_i(t)\ge v_{\textup{lb},l}(t)$ for all nodes $i$ in any cell $\mathcal{N}_l$. The proof of $v_i(t)\le v_{\textup{ub},l}(t)$ follows analogously. First, we define the curing rates $\tilde{\delta}_{\textup{max},i}$ by 
\begin{align}\label{hiuouhqwer}
\tilde{\delta}_{\textup{max},i} = \delta_{\textup{max},l}
\end{align}
for all nodes $i$ in any cell $\mathcal{N}_p$. Thus, (\ref{delta_ub_def}) implies that $\tilde{\delta}_{\textup{max},i}\ge \delta_i$ for all nodes $i=1, ..., N$. 

\begin{lemma}\label{lemma:beta_lb_beta}
For all nodes $i,j$, there are infection rates $\tilde{\beta}_{ij}$, which satisfy $\tilde{\beta}_{ij}\le \beta_{ij}$ and 
\begin{align} \label{kjnergergeg}
\sum_{j\in\mathcal{N}_l} \tilde{\beta}_{ij} = d_{\textup{min}, pl}
\end{align}
for all nodes $i$ in any cell $\mathcal{N}_p$ and all cells $\mathcal{N}_l$.
\end{lemma}
\begin{proof}
With the definition of the lower bound $d_{\textup{min},pl}$ in (\ref{d_lowwwww}), we obtain that (\ref{kjnergergeg}) is satisfied if
\begin{align} \label{asdasdasssssaaa}
\sum_{j\in\mathcal{N}_l} \tilde{\beta}_{ij} = \underset{i\in\mathcal{N}_p}{\operatorname{min}} ~ \sum_{k \in \mathcal{N}_l} \beta_{ik}.
\end{align}
Denote the difference of the infection rates by $\varepsilon_{ij} = \beta_{ij} - \tilde{\beta}_{ij}$. Thus, $\tilde{\beta}_{ij}\le \beta_{ij}$ and $\tilde{\beta}_{ij}\ge 0$ holds if and only if $0\le \varepsilon_{ij}\le \beta_{ij}$. We obtain from (\ref{asdasdasssssaaa}) that the differences $\varepsilon_{ij}$ must satisfy
\begin{align}
\sum_{j\in\mathcal{N}_l} \beta_{ij} - \sum_{j\in\mathcal{N}_l} \varepsilon_{ij} = \underset{i\in\mathcal{N}_p}{\operatorname{min}} ~ \sum_{k \in \mathcal{N}_l} \beta_{ik},
\end{align}
which yields that
\begin{align}\label{kjntgrtgtrrrr}
\sum_{j\in\mathcal{N}_l} \varepsilon_{ij} = \sum_{j\in\mathcal{N}_l} \beta_{ij} - \underset{i\in\mathcal{N}_p}{\operatorname{min}} ~ \sum_{k \in \mathcal{N}_l} \beta_{ik}.
\end{align}
To complete the proof, we must show that there exist some $\varepsilon_{ij}\in [0, \beta_{ij}]$ that solve (\ref{kjntgrtgtrrrr}). Since 
\begin{align}
 \sum_{j\in\mathcal{N}_l} \beta_{ij} \ge \underset{i\in\mathcal{N}_p}{\operatorname{min}} ~ \sum_{k \in \mathcal{N}_l} \beta_{ik}
\end{align}
and $\beta_{ij}\ge 0$, the right side of (\ref{kjntgrtgtrrrr}) is some value in $[0, \sum_{j\in\mathcal{N}_l} \beta_{ij}]$. Since the feasible values of the infection rate differences $\varepsilon_{ij}$ are in the interval $[0, \beta_{ij}]$, the left side of (\ref{kjntgrtgtrrrr}) may attain an arbitrary value in $[0, \sum_{j\in\mathcal{N}_l} \beta_{ij}]$. Thus, there are some infection rate differences $\varepsilon_{ij}\in [0, \beta_{ij}]$ that solve (\ref{kjntgrtgtrrrr}), which completes the proof.
\end{proof}
Lemma~\ref{lemma:beta_lb_beta} states the existence of an $N\times N$ matrix $\tilde{B}_{\textup{min}}$ whose elements $\tilde{\beta}_{\text{min},ij}$ satisfy $\tilde{\beta}_{ij}\le \beta_{ij}$ and (\ref{kjnergergeg}). Thus, $\pi$ is an equitable partition of the matrix $\tilde{B}_{\textup{min}}$. We define the $N\times 1$ viral state $\tilde{v}_{\textup{lb}}(t)$ as
\begin{align}\label{kjbaergergea}
	\frac{d \tilde{v}_{\textup{lb}}(t)}{d t } & = - \operatorname{diag}\left(\tilde{\delta}_{\textup{max},1},..., \tilde{\delta}_{\textup{max},N}\right) \tilde{v}_{\textup{lb}}(t) + \textup{\textrm{diag}}\left( u - \tilde{v}_{\textup{lb}}(t)\right) \tilde{B}_{\textup{min}}  \tilde{v}_{\textup{lb}}(t)
\end{align}	
with the initial viral state
\begin{align}
\tilde{v}_{\textup{lb},i}(0)= \underset{j\in\mathcal{N}_p}{\operatorname{min}} ~ v_j(0)
\end{align}
for all nodes $i$ in any cell $\mathcal{N}_p$. Since $\tilde{v}_{\textup{lb},i}(0)\le v_i(0)$, $\tilde{\delta}_{\textup{max},i}\ge \delta_i$ and $\tilde{\beta}_{\text{min},ij}  \le \beta_{ij}$ for all nodes $i,j$, Theorem~\ref{theorem:NIMFA_bounds} yields that $\tilde{v}_{\textup{lb},i}(t)\le v_i(t)$ for every node $i$ at every time $t$. Furthermore, Theorem~\ref{theorem:equitableOriginal} yields that the $N$-dimensional dynamics of the viral state $\tilde{v}_\textup{lb}(t)$ in (\ref{kjbaergergea}) can be reduced to the $r$-dimensional dynamics of the reduced-size viral state $v_{\textup{lb}}(t)$ in (\ref{wrfaergaergaergrea}), which completes the proof.
	
\end{document}